\documentclass[conference,letterpaper]{IEEEtran}

\usepackage{epsfig}
\usepackage{times}
\usepackage{thmtools, thm-restate}
\usepackage{float}
\usepackage{afterpage}
\usepackage[cmex10]{amsmath}
\usepackage{amstext}
\usepackage{amssymb,bm}
\usepackage{latexsym}
\usepackage{color}
\usepackage{graphicx}
\usepackage{amsmath}
\usepackage{amsthm}
\usepackage{graphicx}
\usepackage[center]{caption}
\usepackage{pstricks}
\usepackage{subfigure}
\usepackage{tikz}
\usepackage{booktabs}
\usepackage{multicol}
\usepackage{lipsum}% dummy text
\usepackage{dblfloatfix}
\usepackage{rotating}
\usepackage{multirow}
\usepackage{tabularx}
\usepackage{pbox}

\allowdisplaybreaks
\newtheorem{thm}{Theorem}%[section]

\newtheorem{lem}[thm]{Lemma}

\newtheorem{rem}{Remark}
\newtheorem{defin}{Definition}

\begin{document}
	
\title{On Secure Capacity of Multiple Unicast Traffic over Separable Networks}
\author{
	\IEEEauthorblockN{Gaurav Kumar Agarwal$^\star$, Martina Cardone$^\dagger$, Christina Fragouli$^\star$ }
	$^\star$ UCLA, CA 90095, USA,
	Email: \{gauravagarwal, christina.fragouli\}@ucla.edu \\
	$^\dagger$ University of Minnesota, Minneapolis, MN 55404, USA, Email: cardo089@umn.edu
%	\thanks{The results in this paper were presented in part at the 2016 IEEE International Symposium on Information Theory, at the 2016 IEEE Globecom Workshop, and at the 10th International Conference on Information Theoretic Security.}
}
\maketitle

\begin{abstract}
	This paper studies the problem of information theoretic secure communication when a source has private messages to transmit to $m$ destinations, in the presence of a passive adversary who eavesdrops an unknown set of $k$ edges. The information theoretic secure capacity is derived over unit-edge capacity separable networks, for the cases when $k=1$ and $m$ is arbitrary, or  $m=3$ and $k$ is arbitrary. This is achieved by first showing that there exists a secure polynomial-time code construction that matches an outer bound over two-layer networks, followed by a deterministic mapping between two-layer and arbitrary separable networks. 
\end{abstract}

\section{Introduction}
Today, a large portion of exchanged data over communication networks is {inherently} {\it sensitive and private} (e.g., banking, professional, health). 
Moreover, given the {recent} progress in quantum computing, we can no longer exclusively rely on computational security: we need to explore unconditionally (information theoretic)  secure schemes. In this paper, we present new results for information {theoretic} security {over networks with multiple unicast sessions.} 
%multiple unicast sessions over networks.

We assume that a source has $m$ private messages to send to $m$ destinations over a network modeled as a directed graph with unit capacity edges.
{This communication occurs}
in the presence of  a passive {external} adversary {who} has unbounded computational capabilities (e.g., quantum computer), but limited network presence, {i.e.,} she can wiretap  (an unknown set of) at most $k$ edges of her choice. We {seek to characterize
	%are interested in characterizing 
	the information theoretic secure capacity} for this setup.

Our results apply to the class of {{\em separable} networks that, broadly speaking, are networks that can be partitioned into a number of edge disjoint subnetworks that satisfy certain properties (see Definition~\ref{def:sepNet} in Section~\ref{sec:sys_model})}.
%defined formally in Section 2 (Definition 3).} 
We establish a direct mapping between {the} secure capacity for separable networks, and {the} secure capacity for {two-layer} networks constructed as {follows. The} source is connected to a set of relays via direct edges. These relays are then connected to the $m$ destinations, such that each destination is directly connected to an (arbitrary)  subset of {the} relays. An example of such a {two-layer} network with $6$ relays and $3$ destinations is shown in Fig.~\ref{fig:combi}. 
\begin{figure}[!ht]
	\centering
	\includegraphics[width=0.55\columnwidth]{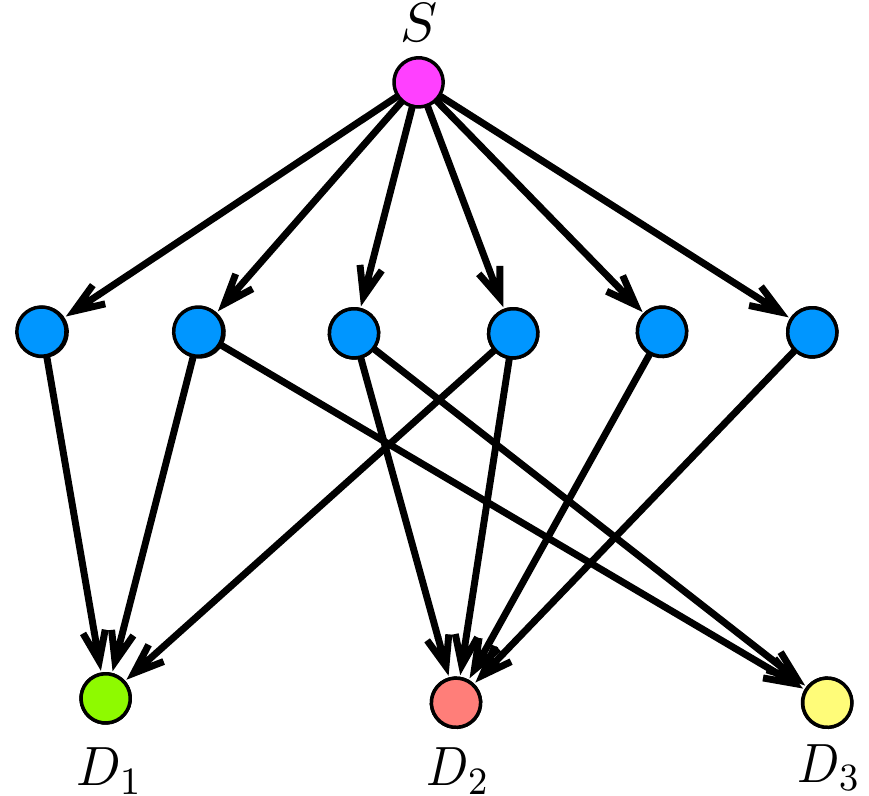}
	\caption{Example of a two-layer network. For $k=1$, the joint scheme achieves the rate triple $(2,2,1)$. This rate triple cannot be achieved by spatially separating the transmissions of the keys and the encoded messages. 
		%		using independent transmissions.} %{\color{red}For $k=1$, the joint scheme can achieve a rate tuple $(2,2,1)$ whereas independent transmission of one common key and the encoded messages can not achieve this rate tuple.}
	}
	\label{fig:combi}
	\vspace{-4mm}
\end{figure}

In~\cite{gauravTIT}, we characterized the secure capacity region for separable networks having $m=2$ destinations. We showed that for $m=2$ it is optimal to use different parts of the network to transmit the keys and the encoded messages.
However, as we also pointed out in~\cite{gauravTIT}, such a scheme is not optimal when $m>2$. We proved this by constructing a joint scheme for two-layer networks that mixes the transmission of keys and encoded messages over the network, and showing that it can achieve higher secure rates than spatially separating the transmissions of the keys and the encoded messages. For instance in Fig.~\ref{fig:combi}, the joint scheme achieves the rate {triple} $(2,2,1)$, which is not possible otherwise.
%, for the network in Fig.~\ref{fig:combi}, we can achieve higher secure rates if we jointly use the network to transmit keys and encoded messages.

%{In~\cite{gauravTIT}, we recently designed a polynomial-time scheme for two-layer networks, and showed that it achieves the secure capacity for the case $m=2$. }
In this paper, we proved that we can leverage the polynomial-time joint scheme in~\cite{gauravTIT} for two-layer networks, to prove capacity results for separable networks  for the following additional cases: 
(i) networks where $m=3$ and $k$ is arbitrary;
(ii) networks where $k=1$ and $m$ is arbitrary;
(iii) networks where $k$ and $m$ are arbitrary, but the network has some special structure in terms of minimum cut.
%The extension to these new cases is not straightforward. 
%For $m =2$, the proof relied on considering key and message transmissions independently, i.e., splitting the network into two disjoint subnetworks, where: (i) one subnetwork is used to multicast the keys to the $m =2$ destinations, and (ii) the other subnetwork is used to transmit the messages encoded with the keys. 
%Unfortunately, as we pointed out in~\cite{gauravTIT}, this approach is not optimal when $m > 2$. For instance, for the network in Fig.~\ref{fig:combi}, we can achieve higher secure rates if we jointly use the network to transmit keys and encoded messages, i.e., multicasting the keys using a portion of the network, and sending the encoded messages using the remaining part does not achieve the secure capacity. 
To prove optimality in these new cases, we needed new proof techniques, that include calculating the dimension of the sum of $m=3$ subspaces in a form that matches a modified outer bound. We also  prove that the secure capacity region of any separable network can be characterized from the secure capacity region of the corresponding two-layer network, referred {to} as the child two-layer network. In particular, we 
provide a deterministic mapping from a secure scheme for the child two-layer network to a secure scheme for the corresponding separable network.
We note that for $m=2$ every network is separable {\cite{GauravICITS2017}}; however this is no longer the case for $m\geq 3$ {\cite{gauravTIT}.}

%\begin{itemize}
%	\item Networks where the number of destinations are three and the number of edges adversary eavesdrop is arbitrary.
%	\item Networks where the adversary eavesdrops only one edge i.e., $k=1$ and the number of destinations are arbitrary.
%	\item Networks where there the number of destinations as well as the edges eavesdropped by the adversary are arbitrary, but the networks have a special class of topology. Under this topology, we require $\mathcal{M}_i \cap\mathcal{M}_j \geq K,  \ \forall i, j$, i.e., the set of relays any pair to destinations are connected to, have at least $k$ in common.      
%\end{itemize} 

\noindent{\bf{Related Work.}} 
Shannon~\cite{shannon1949communication} proved that the one-time pad can provide perfect information theoretic security {with} pre-shared keys. 
For degraded point-to-point channels, Wyner~\cite{Wyner} showed {that} information theoretic security {can be achieved} without pre-shared keys. 
With feedback, {Maurer~\cite{Maurer}} {proved} that secure communication is possible, even when the adversary has a channel {of better quality} than the legitimate receiver. Multicast traffic over networks of unit capacity edges was {analyzed} by Cai et al. in~\cite{cai2002secure}, {and followed by several other works, such as~\cite{feldman2004capacity, el2007wiretap}}. {In~\cite{cai2002secure},} the information theoretic secure capacity was characterized {for networks where} a source {multicasts} the {same} information to a number of destinations in the presence of a passive {external} adversary eavesdropping any {$k$} edges of her choice.  
% In~\cite{cui2013}, Cui et al. studied networks {with non-uniform} edge capacities {when} the adversary {is} allowed to {eavesdrop} only {some specific} subsets of edges.
%In~\cite{GauravICITS2017}, {we studied} security over multiple unicast traffic, and {we characterized the secure capacity region for networks with single source and $m=2$ destinations.} 
In~\cite{layeredNW}, {the} authors studied adaptive and active attacks and also considered multiple multicast traffic over a layered network structure, {with arbitrary number of layers. However,}
different to this paper, every node {in} one layer is connected to every node {in} the next layer. It therefore follows that, for the case of two layers, our setting encompasses {the one in~\cite{layeredNW}.}
%,**** to add here more.

\noindent{\bf{Paper Organization.}} 
In Section~\ref{sec:sys_model} we define {two-layer {and separable} networks, and {formulate the problem.}
	% of characterizing the secure capacity {region.} 
	%{We also formally define  ``separable'' networks.} 
	In Section~\ref{sec:secScheme}, we review} the secure scheme {proposed in~\cite{gauravTIT}}
%for two-layer networks 
and {in Section~\ref{sec:rate_region}}, we characterize {its achieved rate region.}
%achieved by it. 
In Section~\ref{sec:rate_region} we also show the {mapping between separable and two-layer networks.} 
%and two-layer networks.
{In Section~\ref{sec:CapK1} and Section~\ref{sec:3Dest}, we prove} that the scheme achieves the {secure} capacity when $k=1$ and $m=3$, respectively. 
In Section~\ref{sec:3Dest}, we also {provide} sufficient conditions {for the scheme to be optimal for arbitrary $k$ and~$m$.}
%{that ensure that the scheme is capacity achieving for 
%in two-layer networks with arbitrary values of $k$ and $m$.

%on general networks with arbitrary $k$ and arbitrary number of destinations such that the scheme is capacity achieving.  

%\subsection{Related Work}
%The seminal paper of Ahlswede et al.~\cite{AhlswedeIT2000} showed that network coding achieve the capacity for multiple unicast traffic. Li et al.~\cite{li2003linear} later proved the sufficiency of random linear coding and Jaggi et al.~\cite{jaggi2005polynomial} designed polynomial time deterministic algorithms aimed to achieve the multicast capacity. However, no such results are known for multiple unicast traffic in general and Kamath et al.~\cite{kamath2014two} showed that characterizing the capacity of a general network where two unicast sessions is as hard as characterizing the capacity of a network with general number of unicast sessions.

\section{System Model and Problem Formulation}
\label{sec:sys_model}

\label{sec:setting}
\noindent \textbf{Notation:} Calligraphic letters indicate sets;
$\emptyset$ is the empty set;
$\mathcal{A}_1 \sqcup \mathcal{A}_2$ indicates the disjoint union of $\mathcal{A}_1$ and $\mathcal{A}_2$;
$\mathcal{A}_1 \backslash \mathcal{A}_2$ is $\mathcal{A}_1 \cap \mathcal{A}_2^C$;
%$[n_1 : n_2]$ 
%is the set of integers from $n_1$ to $n_2 \geq n_1$;
$[n]: = \{1,2,\ldots, n\}$; 
$[x]^+ := \max\{0, x\}$ for $x \in \mathbb{R}$.
%for a matrix $A$, $A^T$ is its transpose;
%	$\text{dim} (A)$ is the dimension of the rowspace spanned by rows of $A$;
%{for two vector subspaces $V_1$ and $V_2$, we let $V_1+ V_2:= \{ v_1 + v_2: \ v_1 \in V_1, v_2 \in V_2 \}$.}
%	$\mathbf{0}_{i \times j}$ is the all-zero matrix of dimension $i \times j$;
%	$\mathbf{I}_j$ is the identity matrix of dimension $j$.

%We formally define the combination networks as follows. 
{A two-layer network consists of one source $S$ that wishes to communicate with $m$ destinations, by hopping information through one layer of $t$ relays. As such, a two-layer network
is parameterized by: (i) the integer $t$, which denotes the number of relays in the first layer; (ii) the integer $m$, which indicates the number of destinations in the second layer; (iii) $m$ sets $\mathcal{M}_i, \ i \in [m] $, such that $\mathcal{M}_i \subseteq [t]$, where $\mathcal{M}_i$ contains the indexes of the relays connected to destination $D_i$.}
%and $m$, and $m$ subsets of $\{1,2,\ldots, m\}$ termed $\mathcal{M}_i, \ i \in [m] $. A common source $S$ is connected to $t$ relays enumerated from $1$ to $t$. These relays are then connected to $m$ destinations such that $i^{th}$ destination termed $D_i$ is connected to a subset of relays given by $\mathcal{M}_i, \ i \in [m]$. 
An example of {a two-layer network is shown in Fig.~\ref{fig:combi}, for which} $t=6$, $m=3$, $\mathcal{M}_1 = \{1,2,4\}$, $\mathcal{M}_2 = \{3,4,5,6\}$ and $\mathcal{M}_3 = \{2,3\}$. 

We represent a two-layer wireline network with a directed acyclic graph $\mathcal{G}=(\mathcal{V},\mathcal{E})$, where $\mathcal{V}$ is the set of nodes and $\mathcal{E}$ is the set of edges. The edges represent orthogonal and interference-free communication links, {which are}
%{In particular,} these links are 
discrete noiseless memoryless channels of unit capacity over a common alphabet.
If an edge $e \in \mathcal{E}$ connects a node $i$ to a node $j$, we denote, $\text{tail}(e) = i$ and $\text{head}(e) = j$. 
{$\mathcal{I}(v)$ and $\mathcal{O}(v)$ are the set of all incoming and outgoing edges of node $v$, respectively.}
% and $\mathcal{O}(v)$ is the set of all outgoing edges of node $v$.

Source $S$ has a message $W_i$ for destination $D_i, i \in [m]$. These $m$ messages are assumed to be independent.
Thus, the network consists of multiple unicast traffic, where $m$ unicast sessions take place simultaneously and share the network {resources.
%{In particular,} each message $W_i, i \in {[m]},$ is of $q$-ary entropy rate $R_i$.
A passive} {external eavesdropper Eve is also present} and can wiretap any {$k$} edges of her {choice.} 
%We highlight that Eve is an external eavesdropper, i.e., 
%it is not one of the destinations. 
%
The symbol transmitted over $n$ channel uses on {$e \in \mathcal{E}$} is denoted as $X_e^n$. In addition, for $\mathcal{E}_t \subseteq \mathcal{E}$ we define $X_{\mathcal{E}_t}^n= \{ X_e^n: e \in \mathcal{E}_t\}$. We assume that {$S$} {has} infinite sources of randomness $\Theta$, while the other nodes in the network do not have any randomness. 

Over {this} network, we {seek to reliably communicate (with zero error) the message $W_i, i \in [m]$ to destination $D_i$ so that}
%finding all possible feasible $m$-tuples $(R_1,R_2,\ldots, R_m)$ such that {each $D_i, i \in {[m],}$} reliably decodes the message $W_i$ (with zero error) and 
Eve {receives no information} about the {content of the} messages.  
In particular, we are interested in {ensuring} perfect information theoretic secure communication, {and we aim at characterizing} the secure capacity region, {which is next formally defined.}
{
\begin{defin}[Secure Capacity Region]
	\label{defin:sec_cap}
	A rate $m$-tuple $(R_1, R_2, \ldots, R_m)$ is said to be securely achievable if there exist a block length $n$ with $R_i = \frac{1}{n} H(W_i), \ \forall i \in [m]$ and encoding functions  $f_e,   \forall e \in \mathcal{E}$, over a finite field $\mathbb{F}_q$ with
	\begin{align*}
	X^n_e = 
	\left \{
	\begin{array}{ll}
	f_e \left( W_{{[m]}}, \Theta \right) & \mbox{if} \ \mbox{tail}(e)=  S,
	\\
	f_e \left( \{X^n_\ell : \ell \in \mathcal{I}(tail(e))\}  \right) & \mbox{otherwise},
	\end{array} \enspace 
	\right.
	\end{align*}
	{such that each} destination $D_i$ can reliably decode the message $W_i$ i.e., 
	$H\left( W_i | \{X^n_e: e \in \mathcal{I}(D_i)\}\right) = 0, \ \forall i \in [m]$. 

%Moreover, 
{We also} require perfect secrecy, i.e., $
		I \left(W_{{[m]}} ; X^n_{\mathcal{E}_{\mathcal{Z}}}\right)  = 0, \ \forall \ \mathcal{E}_{\mathcal{Z}} \subseteq \mathcal{E} \ \text{such that} \ |\mathcal{E}_{\mathcal{Z}}| \leq k.$
		The \textbf{secure capacity region} is
		the closure of all such feasible rate $m$-tuples.
\end{defin}
{In order to prove that our designed scheme meets the perfect secrecy requirement in Definition~\ref{defin:sec_cap},}
%For the perfect secrecy requirement in Definition~\ref{defin:sec_cap}, 
we will use the {``matrix rank''} condition on perfect secrecy proved in~\cite[Lemma 3.1]{CaiSecCond}.

{We now provide a couple of definitions that will be used in the remaining part of the paper, and we state two remarks that highlight some properties of the networks of interest.}

\begin{defin} [Min-Cut]
	\label{def:MinCut}
	We denote by $M_\mathcal{A}$ the capacity of the min-cut between the source $S$ and the set of destinations $D_\mathcal{A}:=\{D_i,\ i \in \mathcal{A} \}$, and {refer to it as} the
		%Throughout the paper this will be referred by 
		min-cut capacity. 
\end{defin}

\begin{defin}[{Separable Graph}]
	\label{def:sepNet}
	A graph $\mathcal{G} = (\mathcal{V}, \mathcal{E})$ with a  source and $m$ destinations is said to be \textbf{separable} if it can be partitioned into $2^m-1$ edge disjoint graphs (graphs with empty edge sets are also allowed). In particular, these graphs are denoted as $\mathcal{G}_{\mathcal{J}}^\prime = (\mathcal{V},{\mathcal{E}_{\mathcal{J}}^\prime}),  {{\mathcal{J}} \subseteq [m], \mathcal{J} \neq {\emptyset}}$ and are such that ${\mathcal{E}_{\mathcal{J}}^\prime} \subseteq \mathcal{E}$ and 
		$\mathcal{E}_\mathcal{J}^\prime \cap \mathcal{E}_\mathcal{L}^\prime = \emptyset, \ \forall \mathcal{J} \neq \mathcal{L} \subseteq[m]$. {Moreover, their min-cut capacities satisfy} the following condition
		\begin{align}
		M_{\mathcal{A}} & = \sum_{\substack{\mathcal{J}\subseteq {[m]} \\ \mathcal{J} \cap \mathcal{A} \neq \emptyset} } M_{\mathcal{J}}^\prime, \ \forall \mathcal{A} \subseteq {[m]} \label{eq:star_mincuts},
		\end{align}
		where, for {$\mathcal{G}$,} $M_{\mathcal{A}}$ is defined in Definition~\ref{def:MinCut}, and the graph $\mathcal{G}_\mathcal{J}^\prime$ has the following min-cut capacities: (i) $M_\mathcal{J}^{\prime}$ from the source $S$ to any non-empty subset of destinations in $\mathcal{J}$, and (ii) zero from the source $S$ to the set of destinations $\{D_i: i \in [m]\setminus \mathcal{J}\}$.
\end{defin}
%\begin{defin} [Min-cut]
%	\label{def:MinCut}
%	{A \textbf{cut} is} an edge set $\mathcal{E}_\mathcal{A} \subseteq \mathcal{E}$, which separates the source $S$ {from} a set of destinations $D_\mathcal{A}:=\{D_i,\ i \in \mathcal{A} \}$. {In a network with unit capacity edges},
%	{the minimum cut or \textbf{min-cut} is a cut that has the} minimum number of edges. 
%	{ We denote by $M_\mathcal{A}$ the capacity of {the min-cut} between {the source} $S$ and the set of destinations $D_\mathcal{A}$, {and call it the} 
%		%Throughout the paper this will be referred by 
%		min-cut capacity. }
%\end{defin}

%In Definition~\ref{defin:sec_cap}, we require perfect secrecy, i.e., no matter which (at most) $k$ edges Eve eavesdrops, she does not learn anything about the content of the messages. In particular, throughout the paper, we will use the condition on perfect secrecy proved in~\cite[Lemma 3.1]{CaiSecCond}. 
\begin{rem}
For two-layer networks, we have $M_\mathcal{A} = \left|\cup_{i\in \mathcal{A}} \mathcal{M}_i\right|$. 
%The given structure of combination network results in, $M_\mathcal{A} = \left|\cup_{i\in \mathcal{A}} \mathcal{M}_i\right|$. 
For notational convenience, {we let $M_{\cap \{i,j\}} = |\mathcal{M}_i \cap \mathcal{M}_j|$ and $M_{\cap \{i,\mathcal{A} \}} = |\mathcal{M}_i \cap \left(\cup_{j \in \mathcal{A}} \mathcal{M}_j \right)|$. Moreover, we also assume that $M_{\{i\}} \geq K, \forall i \in [m]$ (otherwise secure communication is not possible) with $M_{\emptyset} := K$ for consistency.}
\end{rem}

%\textbf{Remark:} 
\begin{rem}
The single unicast secure capacity~\cite{cai2002secure} is $R_i = M_{\{i\}} - k,  \forall i \in [m]$. 
However, {when} multiple unicast {sessions share} common network resources, { in general, it is not possible to communicate at rate $R_i = M_{\{i\}} - k, \  i \in [m]$ to all destinations simultaneously {(for any choice of $k$)}. This is also highlighted in the outer bound that we derived in~\cite{GauravICITS2017}, which is contained inside the region given by $R_i = M_{\{i\}} - k,  \forall i \in [m]$.}
%does not generally allow to communicate at $R_i = M_{\{i\}} - k, \  i \in [m]$ to all destinations simultaneously, and the actual secure rate region contains inequalities for all possible sum rates. Even for $k=1$ and $m=2$, it is trivial to create examples where $R_i = M_{\{i\}} - 1$ is not achievable.
\end{rem}
}

\section{{Secure Transmissions Scheme}}
\label{sec:secScheme}
%Achievable Scheme and Lower bound on the Capacity
We here review the secure transmission scheme for two-layer networks that we recently {proposed in~\cite{gauravTIT}. 
%In Section~\ref{sec:rate_region}, we will then derive its achieved rate region. 
%$(R_1, R_2, \ldots, R_m)$.
The source} $S$ encodes the message packets with $k$ random packets and transmits {these packets} 
on its outgoing edges to the $t$ relays.
We can write the received symbols at the $t$ relays~as

\begin{align}
\left[\begin{array}{c}
X_1 \\  \vdots \\ X_t
\end{array}\right] & = \left[ \begin{array}{ccc} & & \\ & M & \\ & &  \end{array} | \begin{array}{ccc} & & \\ & V & \\ & &  \end{array} \right] \left[\begin{array}{c}
W_1 \\  \vdots \\  W_m \\ K
\end{array} \right],
\label{eq:SignRxRel}
\end{align}
where: (i) $W_i, i \in [m]$ is a column vector of $R_i$ message packets for destination $D_i$, (ii) $K$ is a column vector which contains the $k$ random packets, (iii) $M$ is a matrix of dimension $t \times \left(\sum_{i=1}^m R_i\right)$, and (iv) $V$ is a Vandermonde matrix of dimension $t \times k$.
The matrix $V$ is chosen to guarantee security as per~\cite[Lemma 3.1]{CaiSecCond}; {hence,} no matter which $k$ edges Eve wiretaps, she learns nothing about the messages $W_{[m]}$.

Each relay $i \in [t]$ forwards the received symbol $X_i$ in~\eqref{eq:SignRxRel} to the destinations it is connected.
As such, each destination will observe a subset of symbols from {$\{X_1, X_2, \ldots, X_t\}$.
% (depending {on} which of {the} $t$ relays it is connected to). 
Finally,} destination $D_i, i \in [m]$ selects a decoding vector and performs the inner product with $[X_1, X_2, \ldots, X_t]$. 
%In particular,
{This decoding} vector is chosen such that it has two characteristics: (1) it is in the left null space of $V$, i.e., in the right null space of $V^T$; this ensures that each destination is able to cancel out the random packets (encoded with the message packets); (2) it has zeros in the positions corresponding to the relays it is not connected to; this ensures that each destination uses only the symbols that {it observes.}
In other words, all the decoding vectors that $D_i$ can choose belong to the null space of the matrix $V_i$ defined~as
\begin{align}
\label{eq:Vi}
V_i = 
\begin{bmatrix}
V^T \\
C_i
\end{bmatrix},
\end{align} 
where $C_i$ is a matrix of dimension $\bar{t} \times t$, with $\bar{t}$ being the number of relays {to which $D_i$} is not connected to. 
In particular, each row of $C_i$ has all zeros except a one in the position corresponding to a relay to which $D_i$ is not connected to.
For instance, with reference to the network in Fig.~\ref{fig:combi}, we have{
\begin{align*}
& C_1  = \begin{bmatrix} 0 & 0 & 1 & 0 & 0& 0   
\\ 0 & 0 & 0 & 0 & 1 & 0
\\ 0 & 0 & 0 & 0 & 0 & 1
\end{bmatrix}, \
 C_3  = \begin{bmatrix} 1 & 0 & 0 & 0 & 0& 0   
 \\ 0 & 0 & 0 & 1 & 0 & 0
 \\ 0 & 0 & 0 & 0 & 1 & 0
 \\ 0 & 0 & 0 & 0 & 0 & 1
 \end{bmatrix}.
\end{align*}}
\section{Achieved Secure Rate Region} \label{sec:rate_region} In this section, we {first} derive the rate region achieved by the secure scheme {in Section~\ref{sec:secScheme}},
{and then present the mapping between separable and two-layer networks.}
In particular, we have the following lemma, whose proof is in Appendix~\ref{App:RateRegAch}.
%leverage the result in Lemma~\ref{lem:LinInd} to prove the next lemma (see Appendix~\ref{App:RateRegAch} for the proof), which provides the secure rate region achieved by our proposed scheme. 
\begin{lem}
	\label{lem:AchScheme}
	The secure rate region achieved by the proposed scheme is given by
	\begin{align}
	\label{eq:AchScheme}
	0 \leq \textstyle{\sum}_{i \in \mathcal{A}} R_i \leq \text{dim}\left( \sum_{i \in \mathcal{A}} N_i\right ), \ \forall  \mathcal{A} \subseteq [m],
	\end{align}
{where $N_i$ is the right null space of the matrix $V_i$ in~\eqref{eq:Vi}.}
\end{lem}
%In {the remainder of this paper,} 
%In {Section~\ref{sec:CapK1} and Section~\ref{sec:3Dest},} we prove that the secure rate region in~\eqref{eq:AchScheme} is indeed the secure {\it capacity} when there are 
%$m=3$ destinations (and {arbitrary} $k$), and when the eavesdropper wiretaps any $k=1$ edge of her choice (and {arbitrary} $m$).

\subsection{Secure Scheme for any Separable Network}
{We will here first} show that for any separable network, a corresponding two-layer network can be created such that both networks have the same min-cut capacities $M_{\mathcal{A}}$ for all $\mathcal{A} \subseteq [m]$.
%Formally, min-cut capacity between source $S$ and a set of destinations $D_\mathcal{A}$ will be same in both the networks for all $\mathcal{A} \subseteq [m]$. 
We will then show that a secure scheme designed for a two-layer network can be converted to a secure scheme on the corresponding separable network.

By Definition~\ref{def:sepNet},
a separable network $\mathcal{G}$ with $m$ destinations, can be separated into $2^m-1$ networks $\mathcal{G}_\mathcal{J}^\prime, \ \mathcal{J} \subseteq [m], \mathcal{J} \neq \emptyset$ where $\mathcal{G}_\mathcal{J}^\prime$ has min-cut capacity $M^\prime_{\mathcal{J}}$ to every subset of destinations in $\mathcal{J}$. 
To construct the corresponding two-layer network, we use the following iterative procedure: (1) we place the source node $S$ in layer $0$ of our network, and the $m$ destination nodes $D_i, i \in [m]$, in layer $2$ of our network; (2) for each $\mathcal{J} \subseteq [m]$, we add $M^\prime_{\mathcal{J}}$ relays in layer $1$ of our network; (3) for each $\mathcal{J} \subseteq [m]$, we connect: (i) the source in layer $0$ with all the added $M^\prime_{\mathcal{J}}$ relays, and (ii) all the added $M^\prime_{\mathcal{J}}$ relays with the destinations $D_i, i \in \mathcal{J}$ in layer $2$. By following the above {procedure, for} each $\mathcal{A} \subseteq [m]$, the min-cut capacity in the constructed two-layer network is $M_{\mathcal{A}}$ as given in~\eqref{eq:star_mincuts}. As such, the new constructed two-layer network has the same min-cut capacity $M_{\mathcal{A}}$ of the corresponding separable network.
	
	In what follows, we refer to the original separable network as {\em parent} separable network, and to the corresponding two-layer network as {\em child} two-layer network.
	%place the source node at the root and for every graph $\mathcal{G}_\mathcal{J}^\prime$, connect $M^\prime_{\mathcal{J}}$ new intermediate nodes in the first layer from the source. We then connect these $M^\prime_{\mathcal{J}}$ nodes of the first layers to the destinations in set $\mathcal{J}$. It can be verified that this two-layer graph has exact same min-cut capacities.
	We now show that a secure scheme designed for the child two-layer network can be converted to a secure scheme for the corresponding parent separable network.
	Towards this end, we
%Now, lets 
assume that we have a secure scheme for the child two-layer {network as described in~\eqref{eq:SignRxRel}, and proceed as follows.} 
%in Section~\ref{sec:secScheme}, we have
%following encoding matrix for which all the destinations can decode their respective messages:
%\begin{align}
%\label{eq:Scheme2Lay}
%X & = \left[ \begin{array}{cc}M & V \end{array} \right] \left[ \begin{array}{c}
%W \\ K
%\end{array}\right].
%\end{align}
%where: { (i) $X$ is the} vector of symbols transmitted on {the} $M_{[m]}$ edges of the {child two-layer} network, {(ii) $W$ is the messages vector} and $K$ is a vector of $k$ uniform i.i.d. key symbols.
%{Note that the scheme above is such that all the destinations can decode their respective messages, and the adversary does not get any information about the content of the messages.}
%{ To transform the above secure scheme into a secure scheme for} the parent separable network, we 
%{proceed as follows. 
On every graph
	%On every graph partition 
	$\mathcal{G}_\mathcal{J}^\prime$ in the parent separable network, we transmit (multicast) the symbols that were transmitted in the child two-layer network from the {source $S$} in layer $0$ to the set of $M^\prime_\mathcal{J}$ relays in layer $1$ that were added when constructing the child two-layer network for $\mathcal{G}_\mathcal{J}^\prime$.
	Note that this multicast towards all destinations $D_i , i \in \mathcal{J}$, is possible since
%on $M^\prime_\mathcal{J}$ corresponding edges in the first-layer of two-layer network. Since the graph 
$\mathcal{G}_J^\prime$ has min-cut capacity $M^\prime_\mathcal{J}$.
%, it is possible to multicast to every destination in $\mathcal{J}$. 
With such a strategy,
%It is easy to see that 
at the end of the transmissions every destination in the parent separable graph still receives the same set of packets as it would have received in the child two-layer network. Thus, all the destinations can still decode their respective messages. In Appendix~\ref{app:two_layer_seperable_sec} we also prove that {this} scheme satisfies the security condition in~\cite[Lemma 3.1]{CaiSecCond}, and hence it is secure. 
Moreover, since the child { two-layer and the parent separable networks} have {equal} min-cut capacities, they have {the} same outer bound {on} the secure capacity region~\cite{GauravICITS2017}. Thus, an optimal scheme on {a child} two-layer network results in an optimal scheme on the corresponding {parent} separable network.

\section{Secure Capacity for $k=1$}
\label{sec:CapK1}

{In this section, we consider {the case when Eve wiretaps any $k=1$} 
%the eavesdropper who has access to any one 
edge of her choice, and characterize the secure capacity region. In particular, we prove the following theorem.
	\begin{thm}
		\label{thm:CapK1}
		For the {two-layer} network when {Eve wiretaps any $k=1$ edge} of her choice, the secure capacity region is
% given by
		\begin{align}
		\textstyle{\sum}_{i \in \mathcal{A}} R_i \leq M_{\mathcal{A}} - C_{\mathcal{A}}, \  \forall \mathcal{A} \subseteq [m], \label{eq:capK1}
		\end{align}
		{with $C_{\mathcal{A}}$ being} the number of connected components in an undirected graph where: (i) there are $|\mathcal{A}|$ nodes, i.e., one for each $i \in \mathcal{A}$;
		(ii) an edge between node $i$ and node $j$, $\{i,j\}\in \mathcal{A}$, $i \neq j$, exists if $\mathcal{M}_i \cap \mathcal{M}_j \neq \emptyset$.
	\end{thm}}
	\noindent\textbf{Outer Bound:} {The} secure capacity region~\cite{GauravICITS2017} is contained in:
	\begin{align}
	\textstyle{\sum}_{i \in \mathcal{A}} R_i \leq M_{\mathcal{A}} - k, \  \forall \mathcal{A} \subseteq [m]. \label{eq:ob1}
	\end{align}
	We {now} show that the outer bound {in~\eqref{eq:ob1}} can be equivalently written as in~\eqref{eq:capK1}.
	Let ${\mathcal{V}_i,} \ i \in [C_\mathcal{A}]$, represent the set of nodes in the $i$-th component of the graph constructed as explained in Theorem~\ref{thm:CapK1}.
	Then, clearly {$\mathcal{A} = \bigsqcup_{i=1}^{C_{\mathcal{A}}} \mathcal{V}_i$} and 
we can write
	\begin{align*}
	\textstyle{\sum}_{i \in \mathcal{A}} R_i  &=  \textstyle{\sum}_{j=1}^{C_\mathcal{A}} \left( \sum_{i \in {\mathcal{V}_j}} R_i \right)   \stackrel{{\rm{(a)}}}{\leq} \sum_{j=1}^{C_\mathcal{A}} \left ( M_{{\mathcal{V}_j}} - k \right ) 
	\\ & \stackrel{{\rm{(b)}}}{=} M_{{\mathcal{V}_1 \cup \mathcal{V}_2 \cup \ldots \cup \mathcal{V}_{C_{\mathcal{A}}}}}\!\!-\! k C_{\mathcal{A}}
 \stackrel{{\rm{(c)}}}{=} M_{\mathcal{A}} -  C_{\mathcal{A}},
	\end{align*}
	where: (i) the inequality in $\rm{(a)}$ follows by applying~\eqref{eq:ob1} for each set ${\mathcal{V}_i}, i \in [C_{\mathcal{A}}]$, (ii) the equality in $\rm{(b)}$ follows since, by construction, $\mathcal{M}_i \cap \mathcal{M}_j = \emptyset$ for all $i \in {\mathcal{V}_x}$ and $j \in {\mathcal{V}_y}$ with $x \neq y$, and (iii) the equality in $\rm{(c)}$ follows since {$\mathcal{A} = \bigsqcup_{i=1}^{C_{\mathcal{A}}} \mathcal{V}_i$ and $k=1$.}
	Thus,~\eqref{eq:ob1} implies~\eqref{eq:capK1}. Moreover, since $C_\mathcal{A} \geq 1$,~\eqref{eq:capK1} implies~\eqref{eq:ob1}.
	This shows that the rate region in Theorem~\ref{thm:CapK1} is an outer bound on the secure capacity region when $k=1$.
	
	We now consider an example of a {two-layer} network and show how the upper bound {derived above} applies to {it.}
	
	\noindent
	\textbf{Example:} Let $\mathcal{A} = \{2,3,4\}$, and assume that $\mathcal{M}_1 = \{1,2\}$, $\mathcal{M}_2 = \{3,4\}$, $\mathcal{M}_3 = \{4,5,6\}$ and $\mathcal{M}_4 = \{7,8\}$. 
	Then, we {construct an undirected graph such that:} (i) it has $3$ nodes since $|\mathcal{A}|=3$ and (ii) has an edge between node $2$ and node $3$ since $\mathcal{M}_2 \cap \mathcal{M}_3 = \{ 4\} \neq \emptyset$. It therefore follows that this graph has $C_{\mathcal{A}} =2$ components.
	In particular, we have
	\begin{align}
	\textstyle{\sum}_{i \in \mathcal{A}} R_i = \textstyle{\sum}_{i \in {\mathcal{V}_1}} R_i + \textstyle{\sum}_{i \in {\mathcal{V}_2}} R_i \leq M_{\{2,3,4\}} -2 = 4,
\label{eq:exUppek1}
	\end{align}
where ${\mathcal{V}_1} = \{2,3\}$ and ${\mathcal{V}_2} = \{4\}$.
	
%	\begin{figure}
%		\centering
%		\includegraphics[width=150px]{example_fig2.pdf}
%		\caption{Graph construction for $\mathcal{A} = \{2,3,4\}$, with $\mathcal{M}_2 = \{3,4\}$, $\mathcal{M}_3 = \{4,5,6\}$ and $\mathcal{M}_4 = \{7,8\}$.}
%		\label{fig:fig_example2}
%	\end{figure}
	
%		Let $\mathcal{A} = \{1,2,3\}$, and assume that $\mathcal{M}_1 = \{1,2,3\}$, $\mathcal{M}_2 = \{3,4\}$ and $\mathcal{M}_3 = \{5\}$. 
%		Then, we will construct the graph shown in Fig.~\ref{fig:example1}, which: (i) has $3$ nodes since $|\mathcal{A}|=3$ and (ii) has an edge between node $1$ and node $2$ since $\mathcal{M}_1 \cap \mathcal{M}_2 = \{ 3\} \neq \emptyset$. It therefore follows that this graph has $C_{\mathcal{A}} =2$ components.
%		In particular, we have $\mathcal{G}_1 = \{1,2\}$ and $\mathcal{G}_2 = \{3\}$.
%		Thus,
%		\begin{align*}
%		\sum_{i \in \mathcal{A}} R_i = \sum_{i \in \mathcal{G}_1} R_i + \sum_{i \in \mathcal{G}_2} R_i \leq M_{\{1,2,3\}} -2 K.
%		\end{align*}
	
%	\begin{figure}[!ht]
%		\centering
%		\includegraphics[width=120px]{example_fig.pdf}
%		\caption{Graph construction for $\mathcal{A} = \{1,2,3\}$, with $\mathcal{M}_1 = \{1,2,3\}$, $\mathcal{M}_2 = \{3,4\}$ and $\mathcal{M}_3 = \{5\}$.}
%		\label{fig:example1}
%	\end{figure}
%	
	
	\noindent\textbf{Achievable Rate Region:} {We here show that the rate region in Theorem~\ref{thm:CapK1} is achieved by the scheme {described} in Section~\ref{sec:secScheme}. }
	%{\gaurav In this section, we will show that if we select $E$ as the Vandermonde matrix, then the security conditions are met, and the achievable rate region matches the outer bound. In particular, we show that:}
{	In particular, we show
%	\begin{align}
%	\label{eq:CondCap}
%	\text{dim}\left( \textstyle{\sum}_{i \in \mathcal{A}} N_i \right) \geq M_{\mathcal{A}} - C_{\mathcal{A}}, \  \forall \mathcal{A} \subseteq [m],
%	\end{align}
%	{where recall that $\text{dim}\left( \sum_{i \in \mathcal{A}} N_i \right)$ is the secure rate performance of our proposed scheme in Section~\ref{sec:secScheme} (see Lemma~\ref{lem:AchScheme}).}
	%{\gaurav Thus, we} conclude that the outer bound in~\eqref{eq:furtherUpp} is indeed tight and hence it represents the secure capacity region for a combination network with $t$ relays and $m$ destinations when $k=1$. 
	%In the remaining part of this section, we prove that~\eqref{eq:CondCap} holds for the case of $k=1$, i.e., when the eavesdropper wiretaps one edge of her choice.
%	{The condition in~\eqref{eq:CondCap} can be equivalently written as $\forall \mathcal{A} \subseteq [m]$,} 
	\begin{align*}
M_{\mathcal{A}} - C_{\mathcal{A}}  \leq \text{dim}\left( \textstyle{\sum}_{i \in \mathcal{A}} N_i \right)  &\stackrel{{\rm{(a)}}}{=}  \text{dim}\left(  \left( \cap_{i \in \mathcal{A}} V_i \right)^\perp  \right)
\\& = t - \text{dim}\left(\cap_{i \in \mathcal{A}} V_i \right),
\end{align*}
where recall that $\text{dim}\left( \sum_{i \in \mathcal{A}} N_i \right)$ is the secure rate performance of our proposed scheme in Section~\ref{sec:secScheme} (see Lemma~\ref{lem:AchScheme}).
Note that} the equality in ${\rm{(a)}}$ follows by using the property of the dual space and the rank nullity theorem, and $V_i, i \in \mathcal{A}$ is defined in~\eqref{eq:Vi}.
%	 with $E$ being the Vandermonde matrix of dimension $t \times 1$, i.e.,
%	\begin{align*}
%	V_i & =  \left[ \begin{array}{l} \mathbf{1}_t \\  C_i \end{array}\right],
%	\end{align*}
%	where $\mathbf{1}_t$ is a row vector of dimension $t$ of all ones (first {column} of the Vandermonde matrix).
	In other words, we next show {that}
	\begin{align}
	\label{eq:DimInt}
	\forall \mathcal{A} \subseteq [m], \   \text{dim}\left(\cap_{i \in \mathcal{A}} V_i \right)  \leq t - M_{\mathcal{A}} + C_{\mathcal{A}}.
	\end{align}
	Towards this end, we would like to count the number of linearly independent vectors $x \in \mathbb{F}_q^t$ that belong to $\left(\cap_{i \in \mathcal{A}} V_i \right)$.
	%Recall, that for $k=1$ and $i \in \mathcal{A}$, we have 
	%\begin{align*}
	%V_i & =  \left[ \begin{array}{l} \mathbf{1}_t \\  C_i \end{array}\right],
	%\end{align*}

	{We note that,} by our construction: {(i) $V^T$ consists of one row of $t$ ones, and (ii) $C_i$ has} zeros in the positions indexed by $\mathcal{M}_i$. {Hence,} if a vector belongs to $V_i$, then all its components indexed by $\mathcal{M}_i$ have to be the same, {i.e., either they are all zeros, or they are all equal to a multiple of one.}
	Thus, we have $q$ choices to fill these positions indexed by $\mathcal{M}_i$.
%	Let $\mathcal{F}_i$ be the collection of all these vectors with equal components in $\mathcal{M}_i$.
	
	Now, consider $V_j$ with $j \in \mathcal{A}$ and $j \neq i$. 
	By using the same logic as above, if a vector belongs to $V_j$, then all its components indexed by $\mathcal{M}_j$ have to be the same {and we have $q$ choices to fill these. We now need to count the number of such choices that are consistent with the choices made to fill the positions indexed by $\mathcal{M}_i$. }
	
%	
%	; $\mathcal{F}_j$ denotes vectors belonging to $V_j$}).
%	We now need to count the number of {these vectors 
%%that we can select 
%so that} they also belong to $V_i$.
	Towards this end, we consider {two cases:} \\
	$\bullet$ {\bf Case 1:} $\mathcal{M}_i \cap \mathcal{M}_j = \emptyset$. In this case, there is no overlap in the elements indexed by $\mathcal{M}_i$ and $\mathcal{M}_j$ and hence we can select all {the available $q$ choices for the positions indexed by $\mathcal{M}_j$}; \\
	$\bullet$ {\bf Case 2:} $\mathcal{M}_i\ \cap \mathcal{M}_j \neq \emptyset$. There is {an overlap} in the elements indexed by $\mathcal{M}_i$ and $\mathcal{M}_j$. Since we have already fixed the elements indexed by $\mathcal{M}_i$, there is no choice for the elements indexed by $\mathcal{M}_j$ (as all the elements have to be the same).
	
	By iterating the same reasoning as above for all $i \in \mathcal{A}$, we conclude that we can fill all the positions indexed by $\cup_{i \in \mathcal{A}} \mathcal{M}_i$ of a vector $x \in \mathbb{F}_q^t$ and make sure that $x \in \left ( \cap_{i \in \mathcal{A}} V_i \right )$ in $q^{C_{\mathcal{A}}}$ ways. This is because, there are $C_\mathcal{A}$ connected components, and for each of these components we have only $q$ choices to fill the corresponding positions {in the vector $x$ (i.e., the positions that correspond to the relays to which at least one of the destinations inside that component is connected).} Once we fix any position inside a component, in fact all the other positions inside that component have to be the same, and thus we have no more freedom in choosing the other positions. Moreover, the remaining $t - M_{\mathcal{A}}$ {positions of $x$} can be filled with any value in $\mathbb{F}_q$ and for this we have $q^{t - M_{\mathcal{A}}}$ possible choices.
	Therefore, the number of vectors $x \in \mathbb{F}_q^t$ that belong to $\left(\cap_{i \in \mathcal{A}} V_i \right)$ is at most $q^{C_{\mathcal{A}}+t-M_{\mathcal{A}}}$, which implies $
	\forall \mathcal{A} \subseteq [m], \  \text{dim}\left(\cap_{i \in \mathcal{A}} V_i \right)  \leq t - M_\mathcal{A} + C_\mathcal{A}$.
	This proves that the {secure scheme in Section~\ref{sec:secScheme} achieves the} rate region in Theorem~\ref{thm:CapK1}.
% {is} achieved by the {scheme in Section~\ref{sec:secScheme}.}
	We {now} illustrate our method of identifying vectors that belong to $\cap_{i \in \mathcal{A}} V_i$ through {an} example.
	
	\noindent \textbf{Example:} Let $t = 8$, $m = 4$, $\mathcal{M}_1 = \{1,2\}$, $\mathcal{M}_2 = \{3,4\}$, $\mathcal{M}_3 = \{4,5,6\}$ and $\mathcal{M}_4 = \{7,8\}$.
	Let $\mathcal{A} = \{2,3,4\}$.  
	We want to count the number of vectors $x \in \mathbb{F}_q^8$ such that {$x \in V_2 \cap V_3 \cap V_4$.}
% (since $\mathcal{A} = \{2,3,4\}$).
	We use the following iterative procedure:\\
	$\bullet$ {For} $x$ {to} belong to $V_2$ {its} elements in the $3$rd and $4$th positions have to be the same {since} $\mathcal{M}_2 = \{ 3,4\}$. Thus, we have $q$ choices {to fill the $3$rd and $4$th positions}.\\
	$\bullet$ {For} $x$ {to} belong to $V_3$, {its} elements in the $4$th, $5$th and $6$th positions have to be {equal} {since} $\mathcal{M}_3 = \{ 4,5,6\}$. However, the element in the $4$th position has already been fixed in selecting vectors that belong to $V_2$. {Thus, there is no} further choice {in filling the $5$th and $6$th position.}\\
%just repeating the value of the $4$th position in the $5$th and $6$th positions.\\
$\bullet$ {For} $x$ {to} belong to $V_4$, {its} elements in the $7$th and $8$th positions have to be the same {since} $\mathcal{M}_4 = \{ 7,8\}$. Since in the previous two steps, we have not filled yet the elements in these positions, then we have $q$ possible ways to fill the elements in the $7$th and $8$th positions. \\
$\bullet$ Moreover, we can fill the elements in the $1$st and $2$nd positions of $x$ in $q^2$ possible ways. \\
	With the above procedure we get that $\text{dim}\left(\cap_{i \in\{2,3,4\}} V_i \right) = 4$, which is equal to the upper bound that we computed in~\eqref{eq:exUppek1} {for the same example}.
%	\begin{align*}
%	\text{dim}\left(\cap_{i \in\{2,3,4\}} V_i \right) = 4,
%	\end{align*}
%{\blue which is equal to the upper bound that we computed in~\eqref{eq:exUppek1}.}
%	The same number can be computed by constructing the graph corresponding to $\mathcal{A} = \{2,3,4\}$ as explained in Section~\ref{sec:OB}, which is represented in Fig.~\ref{fig:fig_example2}.
%	For this graph, we have $C_{\mathcal{A}}=2$ and from~\eqref{eq:DimInt}
%	\begin{align*}
%	\text{dim}\left(\cap_{i \in\{2,3,4\}} V_i \right) = 8-6+2 =4.
%	\end{align*}

\section{Secure Capacity For $m=3$}
\label{sec:3Dest}
In this section, we consider the case {$m=3$,} and we characterize the secure capacity region {through}
{the theorem~below.
	\begin{thm}
		\label{thm:Cap3Dest}
		For {a two-layer} network with $m=3$ destinations, the secure capacity region is given by
		\begin{align}
		\textstyle{\sum}_{i \in \mathcal{A}} R_i \leq M_{\mathcal{A}} - k, \  \forall \mathcal{A} \subseteq [m]. \label{eq:cap}
		\end{align}
	\end{thm}
	Clearly the {rate} region in~\eqref{eq:cap} is an outer bound on the secure capacity region~\cite{GauravICITS2017} and can
	%In particular, we show that the secure rate region in~\eqref{eq:AchScheme} matches the outer bound given by
	%\begin{align}
	%\sum_{i \in \mathcal{A}} R_i \leq M_{\mathcal{A}} - K, \  \forall \mathcal{A} \subseteq [m], \label{eq:outer_bound1}
	%\end{align}
be equivalently written as}
\begin{align*}
\textstyle{\sum}_{i \in \mathcal{A}} R_i & \leq \!\!\!   \min_{ \small \begin{array}{c}\mathcal{P}:\!\!\bigsqcup\limits_{{\mathcal{Q}} \in \mathcal{P}}\!\!{\mathcal{Q}} \!=\! \mathcal{A} \end{array} } \!\!\left\{ \sum_{ {\mathcal{Q}} \in \mathcal{P}} M_{ \mathcal{Q}} - |\mathcal{P}|k   \right \}, \ \forall \mathcal{A} \subseteq [m], 
%\label{eq:outer_bound_concise}
\end{align*}
{where $\mathcal{P}$ is a partition of $\mathcal{A}$.}
{We will show} that $\forall \mathcal{A} \subseteq [m]$, 
\begin{align}
\text{dim}\left( \textstyle{\sum}_{i \in \mathcal{A}} N_i\right) \!\geq \!\!\! \min_{\small \begin{array}{c}\mathcal{P}:\! \bigsqcup\limits_{{\mathcal{Q}} \in \mathcal{P}}\!\!{\mathcal{Q}} \!=\! \mathcal{A} \end{array} }  \!\!\left\{ \sum_{ {\mathcal{Q}} \in \mathcal{P}} M_{\mathcal{Q}} - |\mathcal{P}|k   \right\}. \label{eq:capacity_bound}
\end{align}
%
%\begin{proof}
We {prove~\eqref{eq:capacity_bound} by considering three different cases.} 
%in three parts - for $|\mathcal{A}| = 1$, for $|\mathcal{A}| = 2$ and for $|\mathcal{A}| = 3$.
	
	%\begin{itemize}
		%\item 
\noindent \textbf{Case 1: $|\mathcal{A}| = 1$, i.e.,} 
		{$\mathcal{A}=\{i\}$. For this case,}
		$V_i$ in~\eqref{eq:Vi} has $ k + t - M_{\{i\}}$ rows. 
		All these rows are linearly independent since:
		(i) the rows of $V^T$ are linearly independent as $V$ is {a} Vandermonde matrix, 
		(ii) $C_i$ is full row rank {by construction,} 
		and (iii) any linear combination of the rows of $V^T$ will have a weight of at least $t-k+1$ (from the Vandermonde property), whereas any linear combination of {the rows of $C_i$} will have a weight of at most $t-M_{\{i\}} \leq t-k$. It therefore follows that, {$\forall i \in [3]$, we have that} $
		\text{dim}(N_i) = t - {\text{dim}(V_i)} 
%\label{eq:Card1} \\
		= t - (k+t- M_{\{i\}}) 
%\nonumber \\
		=  M_{\{i\}} - k$,
		where the first equality follows by using the rank-nullity theorem.
		{Thus,~\eqref{eq:capacity_bound} is satisfied.}
		%
		%\item 

\noindent \textbf{Case 2: $|\mathcal{A}| = 2$, i.e.,} 
		$\mathcal{A}=\{i,j\}$. {$\forall (i, j) \in [3]^2, i \neq j$,} 
%$\text{dim}(N_i + N_j)$
		\begin{align}
		\text{dim}(N_i + N_j) & = \text{dim}(N_i) + \text{dim}(N_j) - \text{dim} (N_i \cap N_j) \nonumber \\
		& = M_{\{i\}} \!+\! M_{\{j\}} \!-\! 2k  \!-\! \text{dim} (N_i \cap N_j),  \label{eq:dim_N_i_N_j}
		\end{align}
{where the second equality follows by using $\text{dim}(N_i)$ derived in Case~1.}
		Thus, we need to compute $\text{dim} (N_i \cap N_j)$. {Note that, by definition, $N_i \cap N_j$ is the right null space of 
			\begin{align*}
			V^\star_{ij} = 
			\begin{bmatrix} V_i \\ V_j	\end{bmatrix}
			\stackrel{\eqref{eq:Vi}}{=}
			\begin{bmatrix}
			V^T \\
			C_i \\
			C_j
			\end{bmatrix}
			=
			\begin{bmatrix}
			V^T\\
			C_{ij}
			\end{bmatrix},
			\end{align*}
			where in the last {equality, 
%follows by removing one copy of the common rows in $C_i$ and $C_j$, i.e., 
$C_{ij}$} is a matrix of dimension $(t -M_{\cap \{i,j\}} )\times t$, with all unique rows. Using a similar argument {as in Case~1} 
%(i.e., any vector in the span of $V^T$ has a minimum weight of $t-k+1$),
%			Since $V^T$ is the generator matrix of {a Maximum Distance Separable (MDS)} code of length $t$ and dimension $k$, 
			the number of linearly independent rows of $V^\star_{ij}$ is
			$\min \{ t, t - M_{\cap \{i,j\}} + k \}$. Thus,
			\begin{align*}
			\text{dim} (N_i \cap N_j) & = t - \min\{t, t - M_{\cap \{i,j\}} + k \} \\
			& = \max \{ 0, M_{\cap \{i,j\}} - k \}= {[M_{\cap \{i,j\}} - k]^+,}
			\end{align*}
			where the first equality follows from the rank-nullity theorem.
			We can now write $\text{dim}(N_i + N_j)$  from~\eqref{eq:dim_N_i_N_j} as ${\text{dim}(N_i + N_j) = \min\left \{ M_{\{i\}} + M_{\{i\}} - 2k, M_{\{i,j\}} - k \right \} }$, and the condition in~\eqref{eq:capacity_bound} is satisfied.
		}
%		\item 

\noindent \textbf{Case 3: $\mathcal{A} = \{1,2,3\}$}. We will {compute} 
		\begin{align}
		\text{dim} (N_1+N_2+N_3) = t - \text{dim}(V_1 \cap V_2 \cap V_3), \label{eq:dimPlusdimInt}
		\end{align}
{that is,} the number of linearly independent {vectors} $x \in \mathbb{F}_q^t$ that belong to $V_1 \cap V_2 \cap V_3$. 
{Similar} to the case $k=1$, we have $t- M_{\{1,2,3\}}$ degrees of freedom to fill the positions of $x$ corresponding to $[t] \setminus \cup_{i \in [3]} \mathcal{M}_i$. We now select a permutation $(i,j,\ell)$ of $(1,2,3)$. 
{In order for $x$ to belong to $V_i$, the positions of $x$} corresponding to $\mathcal{M}_i$ can be filled {with $k$ degrees} of freedom. {This is because: (i) $C_i$ in~\eqref{eq:Vi} has zeros in the positions specified by $\mathcal{M}_i$, and (ii)} $V^T$ has $k$ rows.
%for $x$ to belong in $V_i = \left[ \begin{array}{c} V^T \\ C_i \end{array} \right]$ as $C_i$ has zeros in the position specified by $\mathcal{M}_i$ and $V^T$ has $k$ rows. 
Then, to fill the positions of $x$ specified by $\mathcal{M}_j$ so that $x \in V_j$, we have at most {$[k-M_{\cap \{i,j\}}]^+$ degrees} of freedom. This is because the positions of $x$ corresponding to $\mathcal{M}_i \cap \mathcal{M}_j$ {are already fixed.}
% (when filling the positions of $x$ specified by $\mathcal{M}_i$).
%Once we fix elements of $x$ in positions specified by $\mathcal{M}_i$, from the point of view of $\mathcal{M}_j$, positions in $\mathcal{M}_i \cap \mathcal{M}_j$ are already fixed and thus 
%to fill the {remaining positions specified} by $\mathcal{M}_j$, we have only {$[K-M_{\cap \{i,j\}}]^+$ degrees} of freedom. 
Finally, to fill the {positions of $x$} corresponding to $\mathcal{M}_\ell$ {so} that $x \in V_\ell$, we have at most $ [k - M_{\cap \{\ell, \{i,j\} \}}]^+$ degrees of freedom. This is because the positions of $x$ corresponding to $\mathcal{M}_\ell \cap (\mathcal{M}_i \cup \mathcal{M}_j)$ are already fixed.
%, Similar to the case of $k=1$, the positions corresponding to $[t] \setminus \cup_{i \in [3]} \mathcal{M}_i$ can be filled by $t- M_{\{123\}}$ degrees of freedom.
Thus, we obtain $\text{dim}(V_1 \cap V_2 \cap V_3)  \leq k + {[k- M_{\cap \{i,j\}}]^+}  + {[k - M_{\cap \{\ell, \{i,j\} \}} ]^+} +  t- {M_{\{1,2,3\}}}$, {which when substituted in~\eqref{eq:dimPlusdimInt}, satisfies~\eqref{eq:capacity_bound} (see Appendix~\ref{app:FurtherAnalCase3}).}
%{In Appendix~\ref{app:FurtherAnalCase3}, we further tighten $\text{dim}(V_1 \cap V_2 \cap V_3) $ above and show that, when substituted in~\eqref{eq:dimPlusdimInt}, it satisfies the condition in~\eqref{eq:capacity_bound}. 
This proves Theorem~\ref{thm:Cap3Dest}.
		
{We now conclude this section {with
%by providing sufficient conditions for which the secure scheme in Section~\ref{sec:secScheme} is capacity achieving. In particular, we have 
the following lemma.}
	
%	\end{itemize}    
%\end{proof}

\begin{lem}
\label{lem:GenralCase}
The scheme in Section~\ref{sec:secScheme} achieves the secure capacity region of a two-layer network with arbitrary values of $k$ and $m$ whenever $\mathcal{M}_{\cap \{i,j\}} \geq  k$ for all $(i,j) \in [m]^2, i \neq j$.
\end{lem}}

\begin{proof}{We can} compute $\text{dim} (\cap_{i=1}^m V_i)$ as follows:
	
$\begin{aligned}
\text{dim} (\cap_{i\in \mathcal{A}} V_i)  & \stackrel{{\rm{(a)}}}{\leq} t - M_{\mathcal{A}} + k + {[k -\mathcal{M}_{\cap \{i_1,i_2\}} ]^+} \\  & \quad +  \textstyle\sum_{j=3}^m {[k -\mathcal{M}_{\cap \{i_j, \{i_1, i_2,  \ldots, i_{j-1} \} \}} ]^+} \\
& {\stackrel{{\rm{(b)}}}{\leq}} k + t - M_{\mathcal{A}},
\end{aligned}$
where: { $\rm{(a)}$ follows by extending to arbitrary $m$} the iterative {algorithm for} Case~3 above to select $x \in \cap_{i=1}^m V_m$, {and $\rm{(b)}$ follows since} $
\mathcal{M}_{\cap \{i_j, \{i_1, i_2,  \ldots, i_{j-1} \} \}} \geq \mathcal{M}_{\cap \{i_j,i_{j-1} \}} \geq k$.
{By using the property of the dual space and the rank-nullity theorem, we obtain} $\text{dim} (\sum_{i \in \mathcal{A}} N_i) \geq  M_{\mathcal{A}} - k$, {which} satisfies~\eqref{eq:capacity_bound} $\forall \mathcal{A} \subseteq [m]$. {This proves} Lemma~\ref{lem:GenralCase}.
\end{proof}

\bibliographystyle{IEEEtran}
\bibliography{tit18}

\begin{thebibliography}{10}
\providecommand{\url}[1]{#1}
\csname url@samestyle\endcsname
\providecommand{\newblock}{\relax}
\providecommand{\bibinfo}[2]{#2}
\providecommand{\BIBentrySTDinterwordspacing}{\spaceskip=0pt\relax}
\providecommand{\BIBentryALTinterwordstretchfactor}{4}
\providecommand{\BIBentryALTinterwordspacing}{\spaceskip=\fontdimen2\font plus
\BIBentryALTinterwordstretchfactor\fontdimen3\font minus
  \fontdimen4\font\relax}
\providecommand{\BIBforeignlanguage}[2]{{%
\expandafter\ifx\csname l@#1\endcsname\relax
\typeout{** WARNING: IEEEtran.bst: No hyphenation pattern has been}%
\typeout{** loaded for the language `#1'. Using the pattern for}%
\typeout{** the default language instead.}%
\else
\language=\csname l@#1\endcsname
\fi
#2}}
\providecommand{\BIBdecl}{\relax}
\BIBdecl

\bibitem{gauravTIT}
G.~K. Agarwal, M.~Cardone, and C.~Fragouli, ``On secure network coding for
  multiple unicast traffic,'' \emph{ar{X}iv:1901.02787}, January 2019.

\bibitem{GauravICITS2017}
------, ``Secure network coding for multiple unicast: On the case of single
  source,'' in \emph{Information Theoretic Security}, 2017, pp. 188--207.

\bibitem{shannon1949communication}
C.~E. Shannon, ``Communication theory of secrecy systems,'' \emph{Bell Labs
  Technical Journal}, vol.~28, no.~4, pp. 656--715, 1949.

\bibitem{Wyner}
A.~D. Wyner, ``The wire-tap channel,'' \emph{The Bell System Technical
  Journal}, vol.~54, no.~8, pp. 1355--1387, 1975.

\bibitem{Maurer}
U.~M. Maurer, ``Secret key agreement by public discussion from common
  information,'' \emph{IEEE Trans. Inf. Theory,}, vol.~39, no.~3, pp. 733--742,
  1993.

\bibitem{cai2002secure}
N.~Cai and R.~W. Yeung, ``Secure network coding,'' in \emph{IEEE International
  Symposium on Information Theory}, 2002, pp. 323--.

\bibitem{feldman2004capacity}
J.~Feldman, T.~Malkin, C.~Stein, and R.~Servedio, ``On the capacity of secure
  network coding,'' in \emph{42nd Annual Allerton Conference on Communication,
  Control, and Computing}, 2004, pp. 63--68.

\bibitem{el2007wiretap}
S.~Y. El~Rouayheb and E.~Soljanin, ``On wiretap networks {II},'' in \emph{IEEE
  International Symposium on Information Theory}, 2007, pp. 551--555.

\bibitem{layeredNW}
N.~Cai and M.~Hayashi, ``Secure network code for adaptive and active attacks
  with no-randomness in intermediate nodes,'' \emph{ar{X}iv:1712.09035}.

\bibitem{CaiSecCond}
N.~{Cai} and R.~W. {Yeung}, ``A security condition for multi-source linear
  network coding,'' in \emph{2007 IEEE International Symposium on Information
  Theory}, June 2007, pp. 561--565.

\bibitem{schrijver2003combinatorial}
A.~Schrijver, \emph{Combinatorial optimization: polyhedra and
  efficiency}.\hskip 1em plus 0.5em minus 0.4em\relax Springer Science \&
  Business Media, 2003, vol.~B.

\end{thebibliography}

\clearpage
\appendices

\section{Proof of Lemma~\ref{lem:AchScheme}}
\label{App:RateRegAch}

{
%In what follows we refer to the {right} null space of the matrix $V_i$ as $N_i$.
We let $T$ be a matrix of dimension $\left(\sum_{i=1}^m R_i\right) \times t$ that, for each destination $D_i, i \in [m],$ contains $R_i$ decoding vectors that belong to $N_i$. 
Mathematically, we have
\begin{align}
T & = \left[  \begin{array}{c} ---- d^{(1)}_1 ---- \\ ---- d^{(1)}_2 ---- \\ \vdots \\  ---- d^{(1)}_{R_1} ---- \\ ---- d^{(2)}_1 ---- \\ \vdots \\ ---- d^{(m)}_{R_m} ----
\end{array} \right] \label{eq:decoding_vectors},
\end{align}
where $d^{(i)}_j$ denotes the $j$-th decoding vector {(of length $t$)} selected from the null space $N_i$, with $ i \in [m], j \in [R_i]$.
{Note that, if for all $i \in [m]$, we can select $R_i$ decoding vectors from $N_i$ such that all the $d_j^{(i)}$ in~\eqref{eq:decoding_vectors} are linearly independent (i.e., such that $T$ has a full row rank), then}
it is possible to construct the matrix $M$ in~\eqref{eq:SignRxRel} such that
\begin{align*}
TM = I,
\end{align*}
which ensures that all the destinations are able to correctly decode their intended message as
{
	\begin{align*}
	\begin{bmatrix}
	\hat{W}_1
	\\
	\vdots
	\\
	\hat{W}_m
	\end{bmatrix}
	=
	T
	\begin{bmatrix}
	X_1
	\\
	\vdots
	\\
	X_t
	\end{bmatrix}
	\stackrel{\eqref{eq:SignRxRel}}{=}
	\begin{bmatrix}
	TM & TV
	\end{bmatrix}
	\begin{bmatrix}
	W_1
	\\
	\vdots
	\\
	W_m
	\\
	K
	\end{bmatrix} \\
	= TM 
	\begin{bmatrix}
	W_1
	\\
	\vdots
	\\
	W_m
	\end{bmatrix}
	+TVK
	=
	\begin{bmatrix}
	W_1
	\\
	\vdots
	\\
	W_m
	\end{bmatrix}.
	\end{align*}
}
%In Appendix~\ref{App:AchRate}, 
We propose an iterative algorithm to select $R_i, i \in [m]$ decoding vectors from $N_i$ such that $T$ in~\eqref{eq:decoding_vectors} has indeed a full row rank.
The performance of the proposed algorithm is provided in the following lemma.
%, which is also proved in Appendix~\ref{App:AchRate}.
\begin{lem}
	\label{lem:LinInd}
	For any given permutation $\pi = \left \{ \pi(1), \ldots, \pi(m) \right \}$ of $[m]$, it is possible to select
	\begin{align}
	R_{\pi(i)} = \text{dim} \left(\sum_{j = 1}^i N_{\pi(j)} \right) - \text{dim} \left(\sum_{j =1}^{i-1} N_{\pi(j)} \right), \ i \in [m],  \label{eq:AchCornerPt}
	\end{align}
	vectors from $N_{\pi(i)}$ so that all the $\sum_{i=1}^m R_i$ selected vectors are linearly independent.
\end{lem}
\begin{proof}
	We use an iterative algorithm that, for any permutation $\pi = \{\pi(1),\ldots,\pi(m)\}$ of $[m]$, allows to select $R_{\pi(i)}$ vectors from $N_{\pi(i)}$ (with $R_{\pi(i)}$ being defined in~\eqref{eq:AchCornerPt}) so that all the selected $\sum_{i=1}^m R_i$ vectors are linearly independent.
	We next illustrate the main steps of the proposed algorithm.
	\begin{enumerate}
		\item We select $R_{\pi(1)} = \text{dim}(N_{\pi(1)})$ independent vectors from $N_{\pi(1)}$. Note that one possible choice for this consists of selecting the basis of the subspace $N_{\pi(1)}$. 
		\item Next we would like to select independent vectors from $N_{\pi(2)}$ that are also independent of the $R_{\pi(1)}$ vectors that we selected in the previous step. Towards this end, we note that a basis of the subspace $N_{\pi(1)} + N_{\pi(2)}$ is a subset of the union between a basis of $N_{\pi(1)}$ and a basis of $N_{\pi(2)}$. Therefore, we can keep selecting vectors from a basis of $N_{\pi(2)}$ as long as we select an independent vector. Since there are $\text{dim}(N_{\pi(1)}+N_{\pi(2)})$ independent vectors in a basis of $N_{\pi(1)}+N_{\pi(2)}$, then we can select 
		\begin{align*}
		R_{\pi(2)}=\text{dim}(N_{\pi(1)}+N_{\pi(2)}) - \text{dim}(N_{\pi(1)})
		\end{align*}
		independent vectors from $N_{\pi(2)}$ that are also independent of the $R_{\pi(1)}$ vectors that we selected in the previous step.
		\item Similar to the above step, we now would like to select independent vectors from $N_{\pi(3)}$ that are also independent of the $R_{\pi(1)}+R_{\pi(2)}$ vectors that we selected in the previous two steps. Towards this end, we note that a basis of the subspace $N_{\pi(1)} + N_{\pi(2)}+ N_{\pi(3)}$ is a subset of the union between a basis of $N_{\pi(1)}+N_{\pi(2)}$ and a basis of $N_{\pi(3)}$. Therefore, we can keep selecting vectors from a basis of $N_{\pi(3)}$ as long as we select an independent vector. Since there are $\text{dim}(N_{\pi(1)}+N_{\pi(2)}+N_{\pi(3)})$ independent vectors in a basis of $N_{\pi(1)}+N_{\pi(2)}+N_{\pi(3)}$, then we can select 
		$$R_{\pi(3)}\!=\!\text{dim}(N_{\pi(1)}+N_{\pi(2)}+N_{\pi(3)}) - \text{dim}(N_{\pi(1)}+N_{\pi(2)})$$ independent vectors from $N_{\pi(3)}$ that are also independent of the $R_{\pi(1)}+R_{\pi(2)}$ vectors that we selected in the previous two steps.
		\item We keep using the iterative procedure above for all the elements in $\pi$, and we end up with $\sum_{i=1}^m R_i$ vectors that are linearly independent.
	\end{enumerate}
	This concludes the proof of Lemma~\ref{lem:LinInd}.
\end{proof}

\begin{rem}
	Note that, since there are $m!$ possible permutations of $[m]$, then Lemma~\ref{lem:LinInd} offers $m!$ possible choices for selecting $R_i, i \in [m]$ vectors from $N_i$ so that all the $\sum_{i=1}^m R_i$ selected vectors are linearly independent.
\end{rem}

\begin{rem}
	The result in Lemma~\ref{lem:LinInd} implies that rate $m$-tuple $(R_1,R_2,\ldots,R_m)$, with $R_i, i \in [m]$ being defined in~\eqref{eq:AchCornerPt}, can be securely achieved by our proposed scheme.
\end{rem}
We now leverage the result in Lemma~\ref{lem:LinInd} to prove Lemma~\ref{lem:AchScheme}. We start by noting}
%{\blue
that the rate region in~\eqref{eq:AchScheme} can be expressed as {the following polyhedron}:
\begin{align}
P_f : = \left \{ {R} \in \mathbb{R}^{[m]} :{R} \geq \mathbf{0}, \sum_{i \in \mathcal{A}} {R_i} \leq f(\mathcal{A}), \ \forall \ \mathcal{A} \subseteq {[m]} \right \}, \label{eq:polyMatriod}
\end{align}
where $f(\mathcal{A}) := \text{dim} \left( \sum_{i \in \mathcal{A}} N_i \right)$. We now prove the following lemma, which states that this function $f(\cdot)$ is a {\it non-decreasing} and {\it submodular} function over subsets of $[m]$.
\begin{lem}
	The set function 
	\begin{align*}
	f(\mathcal{A}) := \text{dim} \left( \sum_{i \in \mathcal{A}} N_i \right), \ \forall \mathcal{A} \subseteq [m]
	\end{align*}
	is a non-decreasing and submodular function.
\end{lem}
\begin{proof}
	Let $\mathcal{A} \subset \mathcal{B} \subseteq [m]$, then
	\begin{align*}
	f(\mathcal{B}) & = \text{dim} \left( \sum_{i \in \mathcal{B}} N_i \right)  = \text{dim} \left( \sum_{i \in \mathcal{A}} N_i + \sum_{j \in \mathcal{B} \setminus \mathcal{A}} N_j  \right) \\
	& \geq   \text{dim} \left( \sum_{i \in \mathcal{A}} N_i \right)  = f(\mathcal{A}),
	\end{align*} 
	which proves that the function $f(\cdot)$ is non-decreasing.
	For proving submodularity, consider two subsets $\mathcal{C}, \mathcal{D} \subseteq [m]$. Then, we have
	\begin{align*}
	f(\mathcal{C} \cup \mathcal{D}) & = \text{dim} \left( \sum_{i \in \mathcal{C} \cup \mathcal{D}} N_i \right) =  \text{dim} \left( \sum_{i \in \mathcal{C}} N_i  + \sum_{j \in \mathcal{D}} N_j\right) \\
	& = \text{dim} \left( \sum_{i \in \mathcal{C}} N_i \right) +  \text{dim} \left( \sum_{j \in \mathcal{D}} N_j \right) \\
	& \quad - \text{dim} \left( \left(\sum_{i \in \mathcal{C}} N_i \right) \cap \left(\sum_{j \in \mathcal{D}} N_j \right)  \right) \\
	& \leq \text{dim} \left( \sum_{i \in \mathcal{C}} N_i \right) +  \text{dim} \left( \sum_{j \in \mathcal{D}} N_j \right) \\
	& \quad - \text{dim} \left( \sum_{k \in \mathcal{C} \cap \mathcal{D}} N_k \right) \\
	& = f(\mathcal{C}) + f(\mathcal{D}) - f(\mathcal{C} \cap \mathcal{D}),
	\end{align*}
	which proves that the function $f(\cdot)$ is submodular.
\end{proof}
Since $f(\cdot)$ is a submodular function, then the polyhedron defined in~\eqref{eq:polyMatriod} is the polymatroid associated with $f(\cdot)$. Moreover, since $f(\cdot)$ is also non-decreasing, then the corner points of the polymatroid in~\eqref{eq:polyMatriod} can be found as follows~\cite[Corollary 44.3a]{schrijver2003combinatorial}. Consider a permutation $\pi = \{\pi(1), \ldots, \pi(m)\}$ of $[m]$. Then, by letting ${\mathcal{S}}_\ell=\{\pi(1), \ldots, \pi(\ell)\}$ for $1 \leq \ell \leq m$, we get that the corner points of the polymatroid in~\eqref{eq:polyMatriod} can be written as
\begin{align*}
R_{\pi(\ell)} = f({\mathcal{S}}_{\ell}) - f({\mathcal{S}}_{\ell-1}).
\end{align*}
Note that by using $f(\mathcal{A}) = \text{dim} \left( \sum_{i \in \mathcal{A}} N_i \right)$, the above corner points are precisely those in~\eqref{eq:AchCornerPt} in Lemma~\ref{lem:LinInd}.
Since each rate $m$-tuple $(R_1,R_2,\ldots,R_m)$, with $R_i, i \in [m]$ being defined in~\eqref{eq:AchCornerPt}, can be securely achieved by our proposed scheme, it follows that the secure rate region in~\eqref{eq:AchScheme} can also be achieved by our scheme. This concludes the proof of Lemma~\ref{lem:AchScheme}.
%With $f$ being sub-modular and non-decreasing, the object in~\eqref{eq:polyMatriod} represents a polymatriod over $[m]$ with rank function $f$. The corner points of this polymatriod are given by the following (see Corollary 44.3a in Combinatorial Optimization Volume B by A. Schrijver):
%  
%\begin{align}
%R(\pi(\ell)) = \left\{  \begin{array}{cl}
%f\left( \mathcal{\pi}_{\ell}\right ) -  f\left( \mathcal{\pi}_{\ell-1}\right ) & \ell \in [m^\prime] \\
%0 & \ell \in [m^\prime + 1:m]
%\end{array}  \right., \label{eq:CornerPt}
%\end{align}
%where $m^\prime$ is any integer lesser than equal to $m$,  $ \pi $ is any permutation on $[m]$,  $\mathcal{\pi}_0 = \phi$ and $\mathcal{\pi}_{\ell} = \{ \pi(j) : j \in [\ell] \}$.
%
%It can be seen that the points in~\eqref{eq:AchCornerPt} can be used to form the points in~\eqref{eq:CornerPt}.

%\section{Proof of Lemma~\ref{lem:LinInd}}
%\label{App:AchRate}

\section{Analysis of the dimension of $(V_1 \cap V_2 \cap V_3)$}
\label{app:FurtherAnalCase3}
{From our analysis, we have obtained
		\begin{align}
		\text{dim}(V_1 \cap V_2 \cap V_3) & \leq k + {[k- M_{\cap \{i,j\}}]^+}  \nonumber \\
		& \quad + {[k - M_{\cap \{\ell, \{i,j\} \}} ]^+} \nonumber
\\& \quad +  t- {M_{\{1,2,3\}}}.
\label{eq:dimV1V2V3analysis}
		\end{align}
We now further consider two cases.}
%
%\begin{itemize}
			%\item 

\noindent
{\bf Case 3A:} There {exist $(i, j) \in [3]^2, i \neq j,$} such that  $M_{\cap \{i,j \}} \geq k$. In this {case}, with the permutation $(i, j, \ell)$, {the expression in~\eqref{eq:dimV1V2V3analysis} becomes} 
			\begin{align*}
			&\text{dim}(V_1 \cap V_2 \cap V_3) \leq K \!+\! {[K - M_{\cap \{\ell, \{i,j\} \}} ]^+} \!+\!  t \!-\! {M_{\{1,2,3\}}} \\
			& \quad = t-{M_{\{1,2,3\}}}+ \max \{2k -M_{\cap \{\ell, \{i,j\} \}}, k\}.
			\end{align*}
			{From~\eqref{eq:dimPlusdimInt},} this implies that 
			\begin{align*}
			& {\text{dim}(N_1+N_2+N_3)}
\\& \geq {M_{\{1,2,3\}}} - \max \{2k -M_{\cap \{\ell, \{i,j\} \}}, k \} \\
			& = \min \left \{ {M_{\{1,2,3\}} } -k, {M_{\{\ell\}}} + M_{\{i,j\}} - 2k\right \},  
			\end{align*}
			where {the last equality} follows {since} $M_{\{1,2,3\}} = M_{\{i,j\}} + {M_{\{\ell\}}} -M_{\cap \{\ell, \{i,j\} \}}$. With this, the condition in~\eqref{eq:capacity_bound} is satisfied.
			%\item 

\noindent
{\bf Case 3B:} {We have} $M_{\cap \{i,j \}} < K, \forall {(i,j) \in [3]^2, i \neq j}$. 
In this {case, we compute $\text{dim}(V_1 \cap V_2 \cap V_3)$} as follows: 
{we first fill the positions of $x$} indexed by $\mathcal{M}_1$ with $k$ degrees of freedom, and then fill {the positions of $x$} indexed by $\mathcal{M}_2$ {with} $(k - M_{\cap \{1,2\}})$ degrees of freedom as before. Now, we may have fixed more than $k$ positions {of $x$} corresponding to {indexes in} $\mathcal{M}_3$, which is not feasible. 
If that is the case, we backtrack {(i.e.,} remove excess degrees of freedom) that we have used {for} filling {positions of $x$} indexed by $\mathcal{M}_2$. Thus, 
\begin{enumerate}
\item If $M_{\cap \{3, \{1,2\} \}} \leq k$, then
\begin{align*}
			& {\text{dim}(V_1 \cap V_2 \cap V_3)} 
\\ \leq &  t - M_{\{1,2,3\}} + k + (k -  M_{\cap \{1,2\} \}}) 
\\& +  (k - M_{\cap \{3, \{1,2\} \}} ).
			\end{align*}
{This, from~\eqref{eq:dimPlusdimInt}, implies
\begin{align*}
\text{dim}(N_1+N_2+N_3) \geq M_{\{1\}} + M_{\{2\}} + M_{\{3\}} - 3k,
\end{align*}
which satisfies the condition in~\eqref{eq:capacity_bound}.}
\item If $M_{\cap \{3, \{1,2\} \}} > k$, then
			\begin{align*}
			& {\text{dim}(V_1 \cap V_2 \cap V_3)} 
\\  \leq &  t - M_{\{1,2,3\}} + k + (k -  M_{\cap \{1,2 \}}) - \\ & \min \{ (k -  M_{\cap \{1,2 \}}), (M_{\cap \{3, \{1,2\} \}} - k ) \}.
			\end{align*}
{This, from~\eqref{eq:dimPlusdimInt}, implies
\begin{align*}
&\text{dim}(N_1+N_2+N_3) 
\\ \geq & \min \left \{M_{\{1,2,3\}} - k, M_{\{1\}} + M_{\{2\}} + M_{\{3\}} - 3k \right \},
\end{align*}
which satisfies the condition in~\eqref{eq:capacity_bound}.}
\end{enumerate}
%the $\text{dim}(V_1 \cap V_2 \cap V_3)$ is upper bounded by:
%			\begin{align*}
%			{\text{dim}(V_1 \cap V_2 \cap V_3)} \leq & t - M_{\{1,2,3\}} + K + (K -  M_{\cap \{1,2\} \}}) + \\ & (K - M_{\cap \{3, \{1,2\} \}} ),
%			\end{align*}
%			if $M_{\cap \{3, \{1,2\} \}} \leq K$, and 
%			\begin{align*}
%			{\text{dim}(V_1 \cap V_2 \cap V_3)} & \leq  t - M_{\{1,2,3\}} + K + (K -  M_{\cap \{1,2 \}}) - \\ & \min((K -  M_{\cap \{1,2 \}}), (M_{\cap \{3, \{1,2\} \}} - K ) ),
%			\end{align*}
%			if $M_{\cap \{3, \{1,2\} \}} > K$.
			
%			This implies that $\text{dim}(V_1 \cap V_2 \cap V_3)$ is maximum of following quantities: $t + 3K - M_1 - M_2 - M_3$, $t + K - M_{\{1,2,3\}}$.
%			
%			We get, $\text{dim}(N_1+N_2+N_3)$,
%			\begin{align*}
%			& \geq \min \left \{ M_{\{123\}}-K, M_1 + M_2 + M_3 - 3K\right \},  
%			\end{align*}
%			that also satisfies the condition in~\eqref{eq:capacity_bound}.
%		\end{itemize}

{
	\section{Proof of Security: Separable Networks}
	\label{app:two_layer_seperable_sec}

	{In this section we will show that the scheme described for any separable network is also secure.}
	%Only thing that remains is the proof of security.
{Towards this end, we start by noting that the only property of $V$ that we have used in our scheme for two-layer networks is the Maximum Distance Separable (MDS) property (i.e., any $k$ rows of $V$ are linearly independent).
This implies that, even if in~\eqref{eq:SignRxRel} we select a random matrix $\tilde{V}$ instead of $V$, with high probability (close to $1$ for large field size) we will have a secure scheme for the child two-layer network. }

	Let $Y$ be the collection of the symbols transmitted {(multicast) on the parent separable network, as described above.} 
	Since {multicasting} involves network coding, we {have}
	\begin{align}
	Y = \left[ \begin{array}{c} G \end{array} \right] X,
	\end{align}
	where $G$ is {an encoding matrix of dimension $|\mathcal{E}| \times M_{[m]}$ 
		Thus,}
	%This can be written as:
	\begin{align*}
	Y & = \left[ \begin{array}{c} G \end{array} \right]\left[ \begin{array}{cc}M & \tilde{V} \end{array} \right] \left[ \begin{array}{c} W \\ K\end{array}\right]  \\
	& = \left[ \begin{array}{cc}GM & G\tilde{V} \end{array} \right] \left[ \begin{array}{c} W \\ K\end{array}\right].
	\end{align*}
	{From the security condition in~\cite[Lemma 3.1]{CaiSecCond}, it follows that the scheme above is secure}
	%Following the security condition in Lemma~\ref{lemm:sec_cond}, security follows 
	if we can show that for any choice of $G$, there exists a $\tilde{V}$ such that $\tilde{V}$ is an {MDS matrix (i.e., any $k$ rows of $\tilde{V}$ are linearly independent)
		%(this $\tilde{V}$ will give a corresponding $M$, to use for the two-layer network) 
		and}
	\begin{align*}
	rk\left(\left. \left[ \begin{array}{cc}GM & G\tilde{V} \end{array} \right]\right|_\mathcal{Z}\right) = rk\left(\left. \left[ \begin{array}{cc}G\tilde{V} \end{array} \right]\right|_\mathcal{Z}\right), \ \forall |\mathcal{Z}| \leq k.
	\end{align*} 

	{We} will show that for any choice of $G$ {of size $|\mathcal{E}| \times M_{[m]}$ with $M_{[m]} \geq k$,} there exists a $\tilde{V}$ such that $\tilde{V}$ is an MDS matrix {(i.e., any $k$ rows of $\tilde{V}$ are linearly independent)} and
	\begin{align}
	\label{eq:cond_sec_separable}
	rk\left(\left. \left[ \begin{array}{cc}GM & G\tilde{V} \end{array} \right]\right|_\mathcal{Z}\right) = rk\left(\left. \left[ \begin{array}{cc}G\tilde{V} \end{array} \right]\right|_\mathcal{Z}\right), \ \forall |\mathcal{Z}| \leq k.
	\end{align}
	{We start by noting that
		\begin{align*}
		rk\left(\left. \left[ \begin{array}{cc}G\tilde{V} \end{array} \right]\right|_\mathcal{Z}\right)
		&= rk\left( \left.  G \right|_\mathcal{Z} \cdot \tilde{V}  \right )
		\\& \leq rk\left(\left. \left[ \begin{array}{cc}GM & G\tilde{V} \end{array} \right]\right|_\mathcal{Z}\right) \\
		& = rk\left(\left.  G \right|_\mathcal{Z}  \cdot \left[ \begin{array}{cc}M & \tilde{V} \end{array} \right]\right) 
		\leq rk\left( \left.  G \right|_\mathcal{Z} \right ).
		\end{align*}
		Thus, if we prove that, for all $|\mathcal{Z}| \leq k$,
		\begin{align}
		\label{eq:CondGoodRank}
		rk\left( \left.  G \right|_\mathcal{Z} \cdot \tilde{V}  \right ) = rk\left( \left.  G \right|_\mathcal{Z} \right ),
		\end{align}
		then we also show that~\eqref{eq:cond_sec_separable} holds.
		In what follows, we formally prove that a $\tilde{V}$ such that $\tilde{V}$ is an MDS matrix that satisfies the condition in~\eqref{eq:CondGoodRank} for all $|\mathcal{Z}| \leq k$ can be constructed with a non-zero probability.
		Towards this end, we let $\hat{k} = rk\left( \left.  G \right|_\mathcal{Z} \right )$, where $\hat{k} \leq k$ since $|\mathcal{Z}| \leq k$.  { We define the event 
		\begin{align*}
		A = \left \{rk\left(\left. \left[ \begin{array}{cc} G \end{array} \right]\right|_\mathcal{Z} \tilde{V} \right) = rk\left(\left. \left[ \begin{array}{cc} G \end{array} \right]\right|_\mathcal{Z} \right), \  \forall |\mathcal{Z}| \leq k \right \}.
		\end{align*}
		We have}
	%Matrix $G$ is of size $|\mathcal{E}| \times M_{[m]}$, and we have $M_{[m]} \geq k$. We first note that $rk\left(\left[ \begin{array}{cc}G\tilde{V} \end{array} \right]\big|_\mathcal{Z}\right) = rk\left(\left[ \begin{array}{cc}G \end{array} \right]\big|_\mathcal{Z}\right),  \ \forall |\mathcal{Z}| \leq k$ implies~\eqref{eq:cond_sec_separable}.  Moreover, $rk\left(\left[ \begin{array}{cc}G\tilde{V} \end{array} \right]\big|_\mathcal{Z}\right) = rk\left(\left[ \begin{array}{cc}G \end{array} \right]\big|_\mathcal{Z} \tilde{V}\right)$. We denote $rk\left(\left[ \begin{array}{cc}G \end{array} \right]\big|_\mathcal{Z}\right)$ by $\hat{k} \leq k$. Then, $rk\left(\left[ \begin{array}{cc}G \end{array} \right]\big|_\mathcal{Z} \tilde{V}\right)$ is the dimension of vector space spanned by $k$ random vector picked in a vector space of the dimension $\hat{k}$. With high probability, this will be of dimension $\hat{k}$ and will also maintain MDS property if the field size is sufficiently large.  Formally, following the similar argument as in the proof of case-1, we get
	\begin{align*}
	& \Pr \left\{ A \cap \left\{  \tilde{V} \text{ is MDS} \right\}\right\}  
	\\ & {\stackrel{{\rm{(a)}}}{=} 1 - \Pr\left\{A^c \cup \left\{  \tilde{V} \text{ is not MDS} \right\}\right\} }
	\\ & {\stackrel{{\rm{(b)}}}{\geq} 1 - \underbrace{\Pr \left\{  A^c\right\}}_{P_1}   - \underbrace{\Pr\left\{  \tilde{V} \text{ is not MDS} \right\}}_{P_2},}
	\end{align*}
	{where: (i) the equality in $\rm{(a)}$ follows by using the De Morgan's laws, and (ii) the inequality in $\rm{(b)}$ follows since for two events $A$ and $B$, we have $\Pr(A \cup B)\leq \Pr(A)+\Pr(B)$. We now further upper bound the two probability terms $P_1$ and $P_2$. For $P_1$, we obtain
		\begin{align*}
		P_1 &= \Pr \left\{ rk\left(\left. \left[ \begin{array}{cc} G \end{array} \right]\right|_\mathcal{Z} \tilde{V} \right) = rk\left(\left. \left[ \begin{array}{cc} G \end{array} \right]\right|_\mathcal{Z} \right), \ \forall |\mathcal{Z}| \leq k\right\}^c
		\\ & \stackrel{{\rm{(c)}}}{=} \Pr \left \{ \left( \bigcap_{\mathcal{Z}: |\mathcal{Z}|\leq k} A_{\mathcal{Z}} \right )^c\right \} \\
		& \stackrel{{\rm{(d)}}}{=} \Pr \left \{ \bigcup_{\mathcal{Z}: |\mathcal{Z}|\leq k} (A_{\mathcal{Z}})^c \right \}
		\\& \stackrel{{\rm{(e)}}}{\leq} \sum_{\mathcal{Z}: |\mathcal{Z}|\leq k} \Pr \left \{ (A_{\mathcal{Z}})^c \right \}
		\\& \stackrel{{\rm{(f)}}}{\leq} \left ( \frac{{\rm{e}} |\mathcal{E}|}{k}\right )^k \max_{\mathcal{Z}: |\mathcal{Z}|\leq k} \Pr \left \{ (A_{\mathcal{Z}})^c \right \}
		\\& = \left ( \frac{{\rm{e}} |\mathcal{E}|}{k}\right )^k \max_{\mathcal{Z}: |\mathcal{Z}|\leq k} \left( 1-\Pr \left \{ A_{\mathcal{Z}} \right \} \right )\\ 
		&\stackrel{{\rm{(g)}}}{\leq} \left ( \frac{{\rm{e}} |\mathcal{E}|}{k}\right )^k \max_{\mathcal{Z}: |\mathcal{Z}|\leq k} \left( 1-\Pr \left \{ \hat{A}_{\mathcal{Z}} \right \} \right )\\
		& \stackrel{{\rm{(h)}}}{=} \left ( \frac{{\rm{e}} |\mathcal{E}|}{k}\right )^k  \left( 1-  \prod_{i=0}^{\hat{k}-1} \left(1-\frac{q^i}{q^{\hat{k}}} \right) \right ) \\
		& \stackrel{{\rm{(i)}}}{\leq} \left ( \frac{{\rm{e}} |\mathcal{E}|}{k}\right )^k \left( 1-  \left(1-\frac{1}{q} \right)^k \right ),
		\end{align*}
		where: (i) the equality in $\rm{(c)}$ follows by defining, for a given $\mathcal{Z}$ such that $|\mathcal{Z}|\leq k$, the event 
		\begin{align*}
		A_{\mathcal{Z}} = \left \{rk\left(\left. \left[ \begin{array}{cc} G \end{array} \right]\right|_\mathcal{Z} \tilde{V} \right) = rk\left(\left. \left[ \begin{array}{cc} G \end{array} \right]\right|_\mathcal{Z} \right) \right \},
		\end{align*}
		(ii) the equality in $\rm{(d)}$ follows by using the De Morgan's laws;
		(iii) the inequality in $\rm{(e)}$ follows by using the union bound;
		(iv) the inequality in $\rm{(f)}$ follows since ${{n} \choose {t}} \leq \left ( \frac{{\rm{e}}n}{t} \right)^t$; 
		(v) the inequality in $\rm{(g)}$ follows by defining the event $\hat{A}_{\mathcal{Z}}$ as
		\begin{align*}
		\hat{A}_{\mathcal{Z}} = \left \{rk\left( \hat{G} \hat{V} \right) = rk\left(\left. \left[ \begin{array}{cc} G \end{array} \right]\right|_\mathcal{Z} \right) \right \},
		\end{align*}
		where { $\hat{G}$ is the matrix formed by the $\hat{k} = rk\left( \left.\left[ \begin{array}{c} G \end{array} \right]\right|_\mathcal{Z} \right)$ independent rows of $\left.\left[ \begin{array}{c} G \end{array} \right]\right|_\mathcal{Z}$,} 
		and $\hat{V}$ is formed by the first $\hat{k}$ columns of $\tilde{V}$. Thus, the inequality in $\rm{(g)}$ then follows since $\hat{A}_{\mathcal{Z}} \subseteq A_{\mathcal{Z}}$; 
		(vi) the equality in $\rm{(h)}$ follows due to the following computation. 
		We write 
		\begin{align*}
		 \hat{V} & = \left[ \begin{array}{cccc} v_1 & v_2 & \ldots & v_{\hat{k}} \end{array} \right]  \\   \implies		
		\hat{G} \hat{V} & = \left[ \begin{array}{cccc} \hat{G}v_1 & \hat{G}v_2 & \ldots & \hat{G}v_{\hat{k}} \end{array} \right].
		\end{align*}
		{Note that the matrix $\hat{G} \hat{V}$ is of full rank (equal to $\hat{k}$) if}
		%This will be full rank (i.e., $\hat{k}$) if 
		the only solution to $ \sum_{i=1}^{\hat{k}} c_i \hat{G} {v}_i = 0 $ is $c_i = 0, \forall i \in [\hat{k}]$. 
		Let $\hat{N}$ be the null space of $\hat{G}$, and $\hat{N}^\perp$ be the space such that $\hat{N}^\perp \cap \hat{N} = {\emptyset}$ and $\hat{N}^\perp \cup \hat{N} = \mathbb{F}_q^{M_{[m]}}$. 
		Then, we can write each $v_i, i \in [\hat{k}],$ as the sum of its projection on $\hat{N}$ (say $v_i^{(a)}$) and the residual in $\hat{N}^\perp$ (say $v_i^{(b)}$). 
		This implies that $\hat{G} \hat{V}$ is of full rank if the only solution to $\sum_{i=1}^{\hat{k}} c_i \hat{G} {v}_i^{(b)} = 0 $ is $c_i = 0, \forall i \in [\hat{k}]$ (because $\hat{G} v_i^{(a)} = 0$). 
		Since a random choice of $v_i$ results in a random choice on $v_i^{(b)}$, then the probability of $\hat{G} \hat{V}$ being of full rank is equal to the probability that all the vectors $v_i^{(b)}, \ i \in [\hat{k}]$ are mutually independent in $\hat{N}^\perp$. 
		This probability, since {$\text{dim}(\hat{N}^\perp) = \hat{k}$,} is equal to  $\prod_{i=0}^{\hat{k}-1} \left(1-\frac{q^i}{q^{\hat{k}}} \right)$; finally, (vii) the inequality in $\rm{(i)}$ follows since $i -\hat{k} \leq -1$ for all $i \in [0:\hat{k}-1]$ and $\hat{k} \leq k$.

		For $P_2$, we obtain
		\begin{align*}
		P_2 &= \Pr\left\{  \tilde{V} \text{ is not MDS} \right\} \\
		& \stackrel{{\rm{(j)}}}{=} \Pr \left \{ \left ( \bigcap_{\mathcal{S}: |\mathcal{S}|=k} A_s\right )^c \right \} \\
		& \stackrel{{\rm{(k)}}}{=} \Pr \left \{ \bigcup_{\mathcal{S}: |\mathcal{S}|= k} (A_{\mathcal{S}})^c \right \}
		\\& \stackrel{{\rm{(\ell)}}}{=} {{M_{[m]}} \choose {k}} \Pr \left \{ (A_{\mathcal{S}})^c \right \} \\ 
		& = {{M_{[m]}} \choose {k}} \left (1-\Pr \left \{ A_{\mathcal{S}} \right \} \right) \\
		& \stackrel{{\rm{(m)}}}{=} {{M_{[m]}} \choose {k}} \left (1- \prod_{i=0}^{k-1} \frac{q^k-q^i}{q^k}\right)
		\\ & \stackrel{{\rm{(n)}}}{\leq} {{M_{[m]}} \choose {k}} \left (1- \prod_{i=0}^{k-1} \left( 1-\frac{1}{q}\right)\right) \\
		& = {{M_{[m]}} \choose {k}} \left (1- \left( 1-\frac{1}{q}\right)^k\right),
		\end{align*}
		where: (i) the equality in $\rm{(j)}$ follows by defining, for a given $\mathcal{S}$ such that $|\mathcal{S}|=k$, the event
		\begin{align*}
		A_{\mathcal{S}} = \left\{ \tilde{V}|_{\mathcal{S}} \text{ is full rank} \right \},
		\end{align*}
		(ii) the equality in $\rm{(k)}$ follows by using the De Morgan's laws;
		(iii) the equality in $\rm{(\ell)}$ follows by selecting uniformly at random all the subsets of $k$ rows out of the $M_{[m]}$ rows,
		(iv) the equality in $\rm{(m)}$ follows by counting arguments to ensure that the $k$ selected rows are all independent,
		and (v) the inequality in $\rm{(n)}$ follows since $i -k \leq -1$ for all $i \in [0:k-1]$.
		
		Thus, we obtain
		\begin{align*}
		& \Pr \left\{ A \cap \left\{  \tilde{V} \text{ is MDS} \right\}\right\}   
		\\& \geq 1 - \left ( \frac{{\rm{e}} |\mathcal{E}|}{k}\right )^k \left( 1-  \left(1-\frac{1}{q} \right)^k \right ) \\
		&\quad  - {{M_{[m]}} \choose {k}} \left (1- \left( 1-\frac{1}{q}\right)^k\right)
		\\& >0,
		\end{align*}
		where the last inequality holds for sufficiently large $\mathbb{F}_q$.}
	%
	%
	%\begin{align*}
	%& \geq 1 - {|\mathcal{E}|}^k \max_{|\mathcal{Z}| \leq k}\left(1- Pr \left\{ rk\left(\left[ \begin{array}{cc} G \end{array} \right]\big|_\mathcal{Z} \tilde{V} \right) = rk\left(\left[ \begin{array}{cc} G \end{array} \right]\big|_\mathcal{Z} \right)\right\} \right) - Pr\left\{ \tilde{V} \text{ is not MDS}\right\} \\
	% &\stackrel{(a)}{\geq} 1 - {|\mathcal{E}|}^k \max_{|\mathcal{Z}| \leq k}\left(1 - \left(1-\frac{1}{q} \right)^k  \right)  - Pr\left\{ \tilde{V} \text{ is not MDS}\right\} \\
	% &\geq 1 - {|\mathcal{E}|}^k \max_{|\mathcal{Z}| \leq k}\left(1 - \left(1-\frac{1}{q} \right)^k  \right)  - {M_{[m]} \choose k} \left( 1 - \left(1-\frac{1}{q} \right)^k \right) \\
	% & \stackrel{(b)}{>} 0,
	%\end{align*}
	% where (b) is true for sufficiently large $q$. For (a), lets say $rk\left(\left[ \begin{array}{cc} G \end{array} \right]\big|_\mathcal{Z} \right) = \hat{k} \leq k$ then $rk\left(\left[ \begin{array}{cc} G \end{array} \right]\big|_\mathcal{Z} \tilde{V} \right) = \hat{k}$ with probability at least $ \prod_{j = 1}^{\hat{k}} \left(1 - \frac{q^{t - \hat{k}_1 + j - 1}}{q^t} \right)$ where $t = M_{[m]}$. This follows because picking first $\hat{k}$ columns of $\tilde{V}$ orthogonal to the null space of $\left[ \begin{array}{cc} G \end{array} \right]\big|_\mathcal{Z}$ suffices. 
}

\end{document}